\pdfoutput=1

\documentclass{article}

\usepackage[a4paper,left=4cm]{geometry}	
\usepackage{lmodern,microtype}	
\usepackage{tikz}				
\usepackage[english]{babel}		
\usepackage{color}				
\usepackage{enumitem}

\definecolor{darkblue}{rgb}{0,0,.8}

\usepackage[nobreak]{cite} 

\usepackage{hyperref}
\hypersetup{
    colorlinks,
    citecolor=red,
    filecolor=black,
    linkcolor=darkblue,
    urlcolor=black
}

\usepackage{amssymb,amsthm}
\usepackage{amsmath}
\usepackage{bm}

\usepackage[capitalise,noabbrev]{cleveref}

\newtheorem{theorem}{Theorem}[section]
\newtheorem{lemma}[theorem]{Lemma}
\newtheorem{proposition}[theorem]{Proposition}
\newtheorem{corollary}[theorem]{Corollary}
\newtheorem{conjecture}[theorem]{Conjecture}

\newcommand{\chit}{\protect\raisebox{0.25ex}{$\chi$}}

\renewcommand{\i}{\text{i}}

\newcommand{\Q}{\mathfrak{Q}}
\newcommand{\q}{\mathfrak{q}}
\newcommand{\ket}[1]{|#1\rangle}
\newcommand{\bra}[1]{\langle #1|}

\title{\Large\bf Open spin chains with dynamic lattice supersymmetry}
\author{\normalsize \textsc{Christian Hagendorf} and \textsc{Jean Li\'enardy}
\medskip\\
{\normalsize
  \begin{minipage}{\textwidth}
  \begin{center}
  \textit{
   Universit\'e catholique de Louvain\\
  Institut de Recherche en Math\'ematique et Physique\\
  Chemin du Cyclotron 2, 1348 Louvain-la-Neuve, Belgium} \\
  \bigskip
  \href{mailto:christian.hagendorf@uclouvain.be}{\normalsize 
\texttt{christian.hagendorf@uclouvain.be}},
\href{mailto:jean.lienardy@uclouvain.be}{\normalsize 
\texttt{jean.lienardy@uclouvain.be}}
  \end{center}
  \end{minipage}
}
}

\date{\normalsize{\today}}

\begin{document}
\maketitle

\begin{abstract}
The quantum spin $1/2$ XXZ chain with anisotropy parameter $\Delta=-1/2$ possesses a dynamic supersymmetry on the lattice. This supersymmetry and a generalisation to higher spin are investigated in the case of open spin chains. A family of non-diagonal boundary interactions that are compatible with the lattice supersymmetry and depend on several parameters is constructed. The cohomology of the corresponding supercharges is explicitly computed as a function of the parameters and the length of the chain. For certain specific values of the parameters, this cohomology is shown to be non-trivial. This implies that the spin-chain ground states are supersymmetry singlets. Special scalar products involving an arbitrary number of these supersymmetry singlets for chains of different lengths are exactly computed. As a physical application, the logarithmic bipartite fidelity of the open quantum spin $1/2$ XXZ chain with $\Delta=-1/2$ and special diagonal boundary interactions is determined.
\end{abstract}

\tableofcontents

\section{Introduction}
\label{sec:Introduction}

The spin $1/2$ XXZ Heisenberg chain is arguably one of the most important exactly solvable quantum models of interacting spins in one dimension. One the one hand, its study has inspired the development of many modern techniques of quantum integrability such as the Bethe ansatz \cite{bethe:31}, the quantum-inverse scattering method \cite{korepin:93,maillet:00_2} or the vertex-operator approach \cite{jimbo:94}. On the other hand, it is related to a variety of physically interesting problems, most notably the theory of quantum magnetism \cite{maillet:07,schollwoeck:08}.

In this article, we study the open integrable XXZ Heisenberg chain with arbitrary spin $\ell/2$ \cite{kirillov:87}. We focus on a particular value of the anisotropy parameter where the spin chain exhibits an additional symmetry beyond its integrability: \textit{supersymmetry} \cite{witten:82}. For concreteness, let us consider the familiar case $\ell=1$. For diagonal boundary fields the spin-chain Hamiltonian of a chain of length $L$ is given by
\begin{subequations}
\begin{equation}
  H = -\frac{1}{2}\sum_{j=1}^{L-1} \left(\sigma_j^1\sigma_{j+1}^1 + \sigma_j^2\sigma_{j+1}^2+ \Delta \sigma_j^3\sigma_{j+1}^3\right)+ p \sigma_1^3 + p'\sigma_L^3
  \label{eqn:GenericOpenXXZ}
\end{equation}
where $\sigma^1,\sigma^2$ and $\sigma^3$ denote the Pauli matrices, $\Delta$ the anisotropy parameter, and $p,p'$ the boundary magnetic fields. Fendley and Yang \cite{yang:04} showed that for
\begin{equation}
  \Delta = -\frac{1}{2}, \quad p=p' = -\frac{1}{4}
\end{equation}%
\label{eqn:OpenXXZIntro}%
\end{subequations}%
the Hamiltonian is supersymmetric: up to a constant shift (that we specify later) it can be written as the anticommutator of a \textit{supercharge} and its adjoint. This supercharge is a nilpotent operator that maps states of a chain of length $L$ to states of a chain of length $L+1$: It is dynamic. Dynamic supersymmetry on the lattice has since been observed for many other spin chains \cite{beisert:04,hagendorf:12,hagendorf:13,meidinger:14}, in particular for the periodic and twisted spin $\ell/2$ XXZ Heisenberg chains with a particular value of its anisotropy parameter that depends on $\ell$  \cite{hagendorf:13}. In the following, we refer to this value as the supersymmetric point.

In this article, we determine a multi-parameter family of boundary magnetic fields for the open integrable XXZ Heisenberg chains with spin $\ell/2$ at the supersymmetric point that are compatible with a dynamic lattice supersymmetry. These boundary terms generically are non-diagonal and may differ at the first and last site of the chain. We achieve this through a generalisation of the supercharges found in previous works. Furthermore, we identify all values of the parameters for which the spin-chain Hamiltonians possess so-called supersymmetry singlets \cite{witten:82}. Supersymmetry singlets are special eigenstates of the Hamiltonian that are annihilated by both the supercharge and its adjoint. If they exist then they are automatically ground states and therefore of great physical interest. Their existence is related to the existence of a non-trivial cohomology of the supercharge. We explicitly compute this cohomology.

We use our cohomology results in order to determine sum rules for special scalar products involving the supersymmetry singlets. Specifically, let us denote by $|\psi_L\rangle$ a singlet for a spin chain of length $L$. We consider the overlaps
\begin{equation}
  \langle \psi_L|\left(|\psi_{L_1}\rangle\otimes |\psi_{L_2}\rangle\otimes \cdots \otimes |\psi_{L_m}\rangle\right), \quad L_1+L_2+\cdots+L_m = L.
  \label{eqn:GeneralGSOverlap}
\end{equation}
  These scalar products and their scaling limit are of interest in the field of quantum quenches and quantum entanglement. In particular, the case $m=2$ is related to an entanglement measure called the \textit{bipartite fidelity} \cite{dubail:11}. For one-dimensional quantum critical systems, this quantity has a large-$L$ asymptotic expansion whose first few terms have been predicted by conformal field theory (CFT) techniques \cite{dubail:11,dubail:13}. These CFT predictions have been confirmed by lattice derivations at the leading order in a few cases \cite{dubail:13,weston:11,weston:12}.
  Here, we use the supersymmetry to show that the scalar product in \eqref{eqn:GeneralGSOverlap} can (in a suitable normalisation) be computed from the sole knowledge of a single special component of each involved singlet. For the Hamiltonian \eqref{eqn:OpenXXZIntro}, we provide exact finite-size expressions of these components and therefore an explicit formula for the scalar products as a function of $L_1,\dots,L_m$. This allows us to exactly compute the large-$L$ expansions of the scalar products and show that they match the CFT predictions both at leading \textit{and} subleading orders.

The layout of this article is as follows. In \cref{sec:SUSYSpinChains} we
review the formalism of dynamic lattice supersymmetry for open quantum spin chains. In particular, we discuss the supercharge of the open integrable XXZ Heisenberg chain with spin $\ell/2$ at its supersymmetric point and particular 
 diagonal boundary interactions. We present a new multi-parameter deformation of this supercharge in \cref{sec:DeformedSUSY}. The deformation allows us to identify a family of non-diagonal boundary interactions that is compatible with an exact lattice supersymmetry. The purpose of \cref{sec:E0States} is to determine the values of the deformation parameters for which the supercharges possess supersymmetry singlets. We achieve this by exactly computing the cohomology of the supercharge.
In \cref{sec:SP}, we analyse a number of properties of these supersymmetry singlets. In particular, we find sum rules for the scalar products \eqref{eqn:GeneralGSOverlap}. We analyse them for the supersymmetry singlets of the Hamiltonian \eqref{eqn:OpenXXZIntro} in \cref{sec:SPXXZ}. In particular, we compute their scaling limit and compare our findings to the predictions of conformal field theory. We present our conclusions in \cref{sec:Conclusion}.

\section{Dynamic lattice supersymmetry}

\label{sec:SUSYSpinChains}

In this section, we recall the concept and formalism of dynamic lattice supersymmetry for spin chains
\cite{yang:04,hagendorf:12,hagendorf:13,meidinger:14}. Furthermore, we discuss the supercharges for the open integrable XXZ spin chains with spin $\ell/2$ at their supersymmetric point that we analyse and generalise in this article.

\paragraph{Supercharge and Hamiltonian.} Throughout this article, we consider open quantum spin chains of finite length. We denote by $V^L$ the Hilbert space of a spin chain of length $L$. It is given by the $L$-fold tensor product of single-spin Hilbert spaces $V$:
\begin{equation}
  V^L = \underset{L\,\text{times}}{\underbrace{V\otimes V\otimes \cdots \otimes V}}.
\end{equation}
We focus on the case $ V  = \mathbb C^{\ell+1}$ where $\ell$ is an arbitrary fixed positive integer. We refer to \cite{meidinger:14} for a more general discussion of models where $V$ is a super vector space.

The supercharge $\Q$ of our spin-chain models is a length-increasing operator that maps $V^L$ to $V^{L+1}$ for each $L\geqslant 1$.\footnote{In previous works, the restriction of $\Q$ to $V^L$ was indicated by a subscript: $\Q_L$. We omit these subscripts for the supercharges (and all other operators) to keep the notation as simple as possible.
If needed, we write $\Q:V^L\to V^{L+1}$ in order to emphasise that we consider the action of the supercharge on $V^L$.}
When acting on $V^L$, it is given by the following alternating sum:
\begin{equation}
  \Q = \sum_{j=1}^L (-1)^j \q_j.
  \label{eqn:SCSupercharge}
\end{equation}
Here, the length-increasing operators $\q_j$ are
\begin{equation}
\q_j = \underset{j-1}{\underbrace{1\otimes \cdots \otimes 1}} \otimes  \q\otimes \underset{L-j}{\underbrace{1 \otimes \cdots \otimes 1}}
\end{equation}
where $\q:V\to V\otimes V$ denotes the so-called \textit{local supercharge}. One checks that if $\q$ obeys the relation
\begin{equation}
  (\q\otimes 1-1\otimes \q)\q = 0
  \label{eqn:Coassociativity}
\end{equation}
then the operator $\Q$ is nilpotent,
\begin{equation}
  \Q^2=0,
  \label{eqn:QNP}
\end{equation}
in the sense that the action of  $\Q^2 : V^L \to V^{L+2}$ yields zero on any element of $V^L$, for each $L\geqslant 1$. We refer to \eqref{eqn:Coassociativity} as the coassociativity property.

The canonical (complex) scalar product of the spin-chain Hilbert space allows us to define the adjoint supercharge $\Q^\dagger$. It is a length-decreasing operator that maps $V^L$ to $V^{L-1}$ for each $L\geqslant 2$. We have
$\langle \psi|(\Q^\dagger|\phi\rangle) = (\langle \phi|(\Q |\psi\rangle))^\ast$ for all $|\phi\rangle \in V^{L}, |\psi\rangle \in V^{L-1}$ and each $L\geqslant 2$. It follows from \eqref{eqn:QNP} that the adjoint supercharge is also nilpotent,
\begin{equation}
  (\Q^\dagger)^2=0.
  \label{eqn:QDNP}
\end{equation}
This means that for each $L\geqslant 3$, the application $(\Q^\dagger)^2:V^L\to V^{L-2}$ yields zero on every element of $V^L$.

In supersymmetric quantum mechanics, the Hamiltonian is given by the anticommutator of a supercharge $\Q$ and its adjoint $\Q^\dagger$:
\begin{equation}
  H = \Q\Q^\dagger + \Q^\dagger \Q.
  \label{eqn:DefHFromQ}
\end{equation}
Using the specific supercharge \eqref{eqn:SCSupercharge} we find that the Hamiltonian \eqref{eqn:DefHFromQ} is the sum of a bulk part, describing nearest-neighbour interactions, and boundary terms:
\begin{equation}
  H = \sum_{i=1}^{L-1} h_{i,i+1} + (h_{\text{\rm \tiny B}})_1+(h_{\text{\rm \tiny B}})_L.
  \label{eqn:Hamiltonian}
\end{equation}
Here, $h_{i,i+1}$ denotes the Hamiltonian density $h: V^2\to V^2$, acting on the sites $i$ and $i+1$. In terms of the local supercharge it is given by
\begin{equation}
h = -(1\otimes \q^\dagger)(\q\otimes 1) - (\q^\dagger \otimes 1)(1\otimes \q)+\q\q^\dagger + \frac12\left(\q^\dagger\q \otimes 1 +1 \otimes \q^\dagger \q \right).
\label{eqn:HamiltonianDensity}
\end{equation}
Furthermore, the boundary interaction at the first and last site of the chain is encoded in the operator $h_{\text{\rm \tiny B}}:V\to V$. In terms of the local supercharge, we find
\begin{equation}
  h_{\text{\rm \tiny B}} = \frac{1}{2}\q^\dagger\q.
  \label{eqn:BoundaryHamiltonian}
\end{equation}
In \eqref{eqn:Hamiltonian} the boundary terms are the same at both ends of the spin chain. In \cref{sec:DeformedSUSY}, we show for a specific choice of $\q$ that the definition supercharge can be generalised in order to incorporate unequal boundary terms at both ends without modifying the bulk part.

Unlike $\Q$ and $\Q^\dagger$ the Hamiltonian $H$ is not a length-changing operator. From \eqref{eqn:QNP} and \eqref{eqn:QDNP} follow the following relations:
\begin{equation}
  H \Q = \Q H, \quad H\Q^\dagger = \Q^\dagger H.
  \label{eqn:CommuteHQ}
\end{equation}
Hence, both $\Q$ and $\Q^\dagger$ are formally symmetry operators, and therefore the Hamiltonian $H$ is said to be supersymmetric. The supersymmetry is dynamic in the sense that the Hamiltonians on the left- and right-hand side of \eqref{eqn:CommuteHQ}
act on the Hilbert spaces of spin chains whose lengths differ by one. This supersymmetry leads to special properties of the spectrum of $H$. The construction \eqref{eqn:DefHFromQ} implies that this spectrum is real and non-negative. Furthermore, it follows from the commutation relations \eqref{eqn:CommuteHQ} that the spectra for chains of different length have common eigenvalues. We discuss these properties in detail below in \cref{sec:Cohomology}.

\paragraph{Local supercharges for XXZ chains with arbitrary spin.}
The work of this article is based on a particular local supercharge. It acts on the canonical basis states $|0\rangle,|1\rangle,\dots,|\ell\rangle$ of $V$ according to
\begin{subequations}
\label{eqn:TrigSupercharge}
\begin{equation}
  \q|0\rangle = 0, \quad \text{and} \quad \q|m\rangle = \sum_{k=0}^{m-1}a_{m,k}|k,m-k-1\rangle, \quad m=1,\dots, \ell,
\end{equation}
where we abbreviated $|m_1,m_2\rangle = |m_1\rangle \otimes |m_2\rangle$. The coefficients $a_{m,k}$ are strictly positive real numbers defined for $0\leqslant k < m\leqslant \ell$. They are given by
\begin{equation}
  \label{eqn:Defamk}
  a_{m,k} = \sqrt{\frac{\{m+1\}}{\{m-k\}\{k+1\}}}, \quad
  \quad\{m\} = \frac{q^m-q^{-m}}{q-q^{-1}},
\end{equation}
where $q$ is the root of unity
\begin{equation}
q = e^{\i \pi/(\ell+2)}.
\label{eqn:RootOfUnity}
\end{equation}
\end{subequations}
It follows that the adjoint local supercharge $\q^\dagger$ acts on a basis vector of $V^2$ according to
\begin{equation}
  \label{eqn:ActionAdjointTrigSC}
  \q^\dagger|m_1,m_2\rangle =
  \begin{cases}
    a_{m_1+m_2+1,m_1}|m_1+m_2+1\rangle, &\text{if } m_1+m_2<\ell,\\
    0, &\text{if } m_1+m_2 \geqslant \ell.
  \end{cases}
\end{equation}

The local supercharge \eqref{eqn:TrigSupercharge} was introduced in \cite{hagendorf:13}. It is closely related to the supercharge of the so-called $M_\ell$ models of Fendley, Nienhuis and Schoutens \cite{fendley:03}. These models describe supersymmetric fermions on a one-dimensional lattice with an exclusion constraint that limits the length of connected fermion clusters to $\ell$. Their supercharge splits connected clusters of $m$ fermions into pairs of adjacent clusters of $k$ and $m-k-1$ fermions, for $k=0,\dots,m-1$, with an amplitude $a_{m,k}$. Locally, the $M_\ell$ models are equivalent to the spin chains considered in this article. However, the spin-chain language easily allows us to generalise the supersymmetry and add new features that are rather difficult to implement in the fermion language. One such feature is particle-hole symmetry for the fermions. In fact, the particule-hole transformation is a complicated non-local operation because of the exclusion constraint. In the spin-chain language, it translates to a simple spin-reversal symmetry. Indeed, the special values for the constants $a_{m,k}$ in \eqref{eqn:Defamk} lead to a Hamiltonian density that is invariant under a spin-reversal transformation as we shall see below. Another feature is the introduction of boundary interactions that break the particle-number conservation in the fermion model (without breaking the supersymmetry). In the spin-chain language, this translates to boundary conditions that break the conservation of the magnetisation. Below, we construct a new family of such boundary interactions for the spin-chain models by deforming the action of the supercharge on the boundary sites of the spin chain.

Let us now illustrate the nature of the spin-chain Hamiltonians resulting from \eqref{eqn:TrigSupercharge}.
For $\ell=1$, \eqref{eqn:TrigSupercharge} defines a local supercharge for the well-known spin $1/2$ XXZ chain with anisotropy parameter $\Delta=-1/2$, found by Fendley and Yang \cite{yang:04}:
\begin{equation}
  \q |0\rangle = 0, \quad \q|1\rangle = \ket{0}\otimes\ket{0}.
  \label{eqn:DefQXXZ}
\end{equation}
The corresponding Hamiltonian density and the boundary terms are readily evaluated from \eqref{eqn:HamiltonianDensity} and \eqref{eqn:BoundaryHamiltonian}. The full spin-chain Hamiltonian is obtained from \eqref{eqn:Hamiltonian}. Identifying
\begin{equation}
  |0\rangle = \begin{pmatrix} 1\\ 0 \end{pmatrix}, \quad
  |1\rangle = \begin{pmatrix} 0\\ 1 \end{pmatrix}, 
\end{equation}
it can conveniently be written in terms of the Pauli matrices
\begin{equation}
 \sigma^1
 =\begin{pmatrix}
   0 & 1 \\
   1 & 0
   \end{pmatrix},
   \quad
\sigma^2
 =\begin{pmatrix}
   0 & -\i \\
   \i & 0
   \end{pmatrix},
   \quad
\sigma^3
 =\begin{pmatrix}
   1 & 0 \\
   0 & -1
   \end{pmatrix}.
\end{equation}
Indeed, one finds the XXZ Hamiltonian with anisotropy parameter $\Delta=-1/2$ and diagonal boundary magnetic fields:
\begin{align}
    \label{eqn:XXZOpenChainND}
  H = -\frac{1}{2}\sum_{j=1}^{L-1} \left(\sigma_j^1\sigma_{j+1}^1+\sigma_j^2\sigma_{j+1}^2-\frac{1}{2}\sigma_j^3\sigma_{j+1}^3\right)-\frac{\sigma_1^3+\sigma_L^3}{4}+\frac{3L-1}{4}.
\end{align}
Up to a constant shift, this is the Hamiltonian \eqref{eqn:OpenXXZIntro} that we discussed in \cref{sec:Introduction}. Besides its lattice supersymmetry and its relation to the (open) $M_1$ model of supersymmetric fermions \cite{beccaria:05}, it has a few interesting features. First, it is isospectral to another XXZ chain that describes the so-called Temperley-Lieb stochastic process \cite{nichols:05,degier:05}. The connection to a stochastic process implies in particular that the spectrum of \eqref{eqn:XXZOpenChainND} is non-negative and contains the non-degenerate ground-state eigenvalue $E=0$. In \cref{sec:E0States}, we use the supersymmetry in order to provide new insights into the properties of the corresponding zero-energy state. Second, we note that \eqref{eqn:XXZOpenChainND} belongs to a family of Hamiltonians that have a spectral overlap with the Hamiltonians of the famous quantum-group invariant XXZ spin chains of Pasquier and Saleur's \cite{pasquier:90}. This connection has recently been used in order to proof the reality of the spectra of the Pasquier-Saleur spin chains \cite{morin:16}.

For arbitrary $\ell$, the Hamiltonian density that derives from \eqref{eqn:TrigSupercharge} was explicitly computed in \cite{hagendorf:13}. We write it as follows:
\begin{subequations}
\label{eqn:ExplicitHamDensity}
\begin{equation}
 h = \sum_{m_1,m_2=0}^\ell
  \sum_{n=-M_1}^{M_2}\beta_{m_1,m_2}^n|m_1+n,
  m_2-n\rangle\langle m_1,m_2|,
\end{equation}
where we abbreviate $M_1 = \min(m_1,\ell-m_2)$ and ${{M_2}} = \min(m_2,\ell-m_1)$. The coefficients have the symmetry property $\beta^n_{m_1,m_2} = \beta_{m_2,m_1}^{-n}$. For $n>0$, they are given by
\begin{equation}
  \beta^n_{m_1,m_2} = -\frac{1}{\{n\}}\sqrt{\frac{\{M_1+1\}\{M_2-n+1\}}{\{M_2+1\}\{M_1+n+1\}}}.
\end{equation}
Furthermore, we have
\begin{equation}
  \beta^0_{m_1,m_2} = c_{M_1+1}+c_{M_2+1}, \quad 
    c_m = \sum_{k=1}^m \frac{\{k+1\}-\{k-1\}}{2\{k\}}.
\end{equation}
\end{subequations}
Furtermore, the term that describes the boundary interactions is diagonal for any $\ell$:
\begin{equation}
  \label{eqn:hBND}
  h_{\text{\rm \tiny B}} =\frac{1}{2}\q^\dagger\q= \sum_{m=1}^\ell c_m|m\rangle\langle m|.
\end{equation}
For $\ell=2$, \eqref{eqn:ExplicitHamDensity} yields a Hamiltonian density which coincides, up to a simple unitary transformation, with the Hamiltonian density of the integrable spin-one XXZ chain (of Fateev and Zamolodchikov \cite{zamolodchikov:81}) at its supersymmetric point.
Further investigations in \cite{hagendorf:13} support the conjecture that for arbitrary $\ell$ \eqref{eqn:ExplicitHamDensity} is, up to a simple unitary transformation, the Hamiltonian density of the quantum integrable spin $\ell/2$ XXZ chain at a particular value of its anisotropy parameter.
This observation is consistent with an analysis of the Bethe-ansatz equations \cite{meidinger:14} but remains to be proven. A possible proof could be obtained from Mangazeev's explicit expressions for the $R$-matrices of fused vertex models \cite{mangazeev:14}. We leave the details of this proof to future investigations.

\section{Supercharges and boundary conditions}

\label{sec:DeformedSUSY}

In this section, we generalise the supersymmetry of the open XXZ spin chain at $\Delta=-1/2$, described by the Hamiltonian  \eqref{eqn:XXZOpenChainND}, and its higher-spin analogues at their supersymmetric point. In \cref{sec:OneParamDef} we find a family of local supercharges $\q(y)$, depending non-trivially on a complex parameter $y$, that have the same Hamiltonian density as $\q$ defined in \eqref{eqn:TrigSupercharge}. The resulting boundary terms however depend non-trivially on $y$ and generically are non-diagonal. Using the properties of $\q(y)$, we generalise in \cref{sec:BoundaryConditions} the action of the supercharge at the first and last site of the spin chain. The resulting supersymmetric spin-chain Hamiltonians have unequal boundary terms at both ends of the chain.

\subsection{A one-parameter deformation}
\label{sec:OneParamDef}

In order to construct $\q(y)$, we need to discuss two special local supercharges: the image $\bar \q$ of $\q$ defined in \eqref{eqn:TrigSupercharge} under spin reversal and a so-called local gauge supercharge.

We start our discussion with $\bar \q$ and some of its properties. The spin-reversal operator $R$ on $V^L,\,L\geqslant 1,$ is a linear operator defined by the following action on the canonical basis states $|m_1,m_2,\dots,m_L\rangle=|m_1\rangle \otimes |m_2\rangle \otimes \cdots\otimes |m_L\rangle$:
\begin{equation}
 R|m_1,\dots,m_L\rangle = |\ell-m_1,\dots,\ell-m_L\rangle.
\end{equation}
We define $\bar \q = R \q R$. This operator acts on the basis vectors of $V$ according to
\begin{equation}
  \label{eqn:BarTrigSupercharge}
  \bar \q|\ell\rangle = 0,\quad \text{and} \quad \bar \q |m\rangle = \sum_{k=m+1}^\ell a_{\ell-m,\ell-k}|k, \ell+1+m-k\rangle, \quad m=0,\dots,\ell-1.
\end{equation}
It follows from $R^2=1$ that $\bar \q$ has the coassociativity property \eqref{eqn:Coassociativity}. Furthermore, it was shown in \cite{hagendorf:13} that, for the specific choice \eqref{eqn:Defamk} of the coefficients $a_{m,k}$, the Hamiltonian densities $h$ of $\q$ and $\bar h$ of $\bar \q$ are equal:
\begin{equation}
  \label{eqn:EqualityHHBar}
    h = \bar h.
\end{equation}
For our construction of $\q(y)$, we note furthermore that $\q$ and $\bar \q$ obey a certain anticommutation relation up to boundary terms.
Indeed, for any $|\psi\rangle \in V$ we have the relation \begin{equation}
  \label{eqn:AntiCommutationQQBar}
  \left((-\q \otimes 1+1\otimes \q)\bar \q+(-\bar\q \otimes 1+ 1\otimes \bar\q) \q\right)|\psi\rangle = |\chit\rangle \otimes |\psi\rangle - |\psi\rangle \otimes |\chit\rangle
\end{equation}
where the vector $|\chit\rangle\in V^2$ is given by
\begin{equation}
  |\chit\rangle = \sum_{m=0}^\ell \chit_m |m,\ell-m\rangle, \quad \chit_m = \frac{1}{\{m+1\}}.
  \label{eqn:DefChi}
\end{equation}

Next, let us recall the concept of a local gauge supercharge $\q_\phi$ \cite{hagendorf:13}. It depends on a vector $|\phi\rangle \in V$ and acts on any vector $|\psi\rangle \in V$ according to
  \begin{equation}
  \q_\phi|\psi\rangle = |\phi\rangle \otimes |\psi\rangle + |\psi\rangle \otimes |\phi\rangle.
  \label{eqn:DefGaugeSC}
\end{equation}
The operator $\q_\phi$ does not have the coassociativity property \eqref{eqn:Coassociativity}. However, it obeys a similar relation:\begin{equation}
  (\q_\phi\otimes 1 - 1 \otimes \q_\phi)
\q_\phi|\psi\rangle = |\chit_\phi\rangle \otimes |\psi\rangle - |\psi\rangle \otimes |\chit_\phi\rangle, \quad 
|\chit_\phi\rangle = |\phi\rangle \otimes |\phi\rangle.
\label{eqn:QuasiCoassociativityGauge}
\end{equation}
 
We define the local supercharge $\q(y)$ as a linear combination of $\q$, $\bar \q$ and a local gauge supercharge. The idea is to adjust the gauge term in such a way that $\q(y)$ has the coassociativity property. Let us write
\begin{equation}
  \q(y) = x\left(\q + y^{\ell+2}\bar \q + \q_{\phi(y)}\right), \quad x= \frac{1}{\sqrt{1+|y|^{2(\ell+2)}}}
    \label{eqn:DefLocalQy}
\end{equation}
where $|\phi(y)\rangle\in V$ is the vector characterising the gauge term. Using \eqref{eqn:AntiCommutationQQBar}, \eqref{eqn:DefGaugeSC} and \eqref{eqn:QuasiCoassociativityGauge} one checks that $\q(y)$ satisfies \eqref{eqn:Coassociativity} if $|\phi(y)\rangle$ solves the quadratic equation
\begin{equation}
  \left(\q +y^{\ell+2}\bar \q\right)|\phi(y)\rangle +|\phi(y)\rangle \otimes |\phi(y)\rangle= y^{\ell+2}|\chit\rangle.
  \label{eqn:EqnForPhi}
\end{equation}
This equation is solved by\footnote{The solution is not unique. For instance, $\ell+1$ other solutions are given by $|\phi(q^{2(k+1)}y)\rangle$ with $k=0,\dots,\ell$.}
\begin{equation}
  |\phi(y)\rangle = \sum_{m=0}^\ell \phi_m(y) |m\rangle, \quad \phi_m(y) = -\frac{y^{m+1}}{\sqrt{\{m+1\}}}.
  \label{eqn:DefPhiGauge}
\end{equation}
The equations \eqref{eqn:DefLocalQy} and \eqref{eqn:DefPhiGauge} define the local supercharge $\q(y)$. It has the property that $\q(y=0)=\q$ and $\lim_{y\to \infty} \q(y) =\bar \q$ where the limit is taken along the real axis. The construction implies that $\q(y)$ has a well-defined transformation property under spin reversal. Indeed, one checks that 
\begin{equation}
  R \q(y) R =\left(\frac{y}{|y|}\right)^{\ell+2}\q(y^{-1})
  \label{eqn:ROnQy}
\end{equation}
for non-zero values of $y$. We note that this is compatible with the values taken at $y=0$ and for $y\to \infty$. Furthermore, this equation shows that the local supercharge $\q(y)$ is invariant under spin reversal up to a sign if and only if $y=\pm 1$.

\paragraph{Hamiltonian density.} We will now show that the Hamiltonian density of $\q(y)$ is independent of $y$. To this end, we use a result for local supercharges that are invariant under a parity transformation. The parity operator $P$ on $V^L,\,L\geqslant 1,$ is the linear operator defined by the following action on the canonical basis states:
\begin{equation}
 P|m_1,m_2,\dots,m_L\rangle = |m_L,m_{L-1}\dots,m_1\rangle.
\end{equation}
A local supercharge $\q$ is parity-invariant if $P\q =\q$ on $V$. The following property of parity-invariant local supercharges was shown in \cite{hagendorf:13}:
\begin{lemma}
  \label{lem:ParityProperty}
  Let $\q$ be a parity-invariant local supercharge and $\q_\phi$ any gauge local supercharge, then the Hamiltonian densities of $\q$ and $\q + \q_\phi$ are equal.
\end{lemma}
One checks that the local supercharges $\q$ and $\bar \q$ defined in \eqref{eqn:TrigSupercharge} and \eqref{eqn:BarTrigSupercharge}, respectively, are parity-invariant.
We use this observation to prove the following statement:
\begin{proposition}
\label{prop:HamDensity}
  The Hamiltonian density $h(y)$ of the local supercharge $\q(y)$ defined by \eqref{eqn:DefLocalQy}, \eqref{eqn:DefPhiGauge} is independent of the parameter $y$.
  \begin{proof}
It follows from \cref{lem:ParityProperty} that $h(y)$ is equal to the Hamiltonian density of $x(\q + y^{\ell+2} \bar \q)$. This Hamiltonian density is obtained from  \eqref{eqn:HamiltonianDensity}. In order to evaluate it, we note that $\q$ and $\bar \q$ obey the following relations
\begin{subequations}
\label{eqn:QQBarRelations}
    \begin{align}
      \bar \q\q^\dagger = (1\otimes \q^\dagger)(\bar \q \otimes 1) +(\q^\dagger\otimes 1)(1\otimes \bar \q ),\quad \bar \q ^\dagger \q = 0,\\
      \q \bar \q^\dagger = (1\otimes  \bar \q^\dagger)(\q \otimes 1) +(\bar \q^\dagger\otimes 1)(1\otimes \q ),\quad \q ^\dagger  \bar\q = 0.
    \end{align}
    \end{subequations}%
    These relations follow from a straightforward calculation, using \eqref{eqn:TrigSupercharge} and \eqref{eqn:BarTrigSupercharge}. Combining \eqref{eqn:HamiltonianDensity} with \eqref{eqn:QQBarRelations}, we obtain
    \begin{equation}
      h(y) = x^2(h + |y|^{2(\ell+2)}\bar h),
    \end{equation}
    where $h$ and $\bar h$ denote the Hamiltonian densities of $\q$ and $\bar \q$, respectively. Furthermore, $\bar h= h$ as noted above. Using the value of $x$ given in \eqref{eqn:DefLocalQy}, it follows that $h(y)=h$, which is independent of $y$.
  \end{proof} 
\end{proposition}

We write
\begin{equation}
  \label{eqn:DefQy}
  \Q(y) = \sum_{j=1}^L (-1)^j \q(y)_j
\end{equation}
for the supercharge constructed from $\q(y)$.
The corresponding Hamiltonian $H(y) = \Q(y)\Q(y)^\dagger+\Q(y)^\dagger \Q(y)$ is of the form \eqref{eqn:Hamiltonian}. In the following, we will often write $H$ for $H(y=0)$.
 
\paragraph{Magnetisation.} 
We define the magnetisation operator by the following action on the basis states of $V^L$:
\begin{equation}
  M|m_1,\cdots, m_L\rangle = \left(\frac{\ell L}{2}-\sum_{i=1}^L m_i\right)|m_1,\cdots, m_L\rangle.
\end{equation}
From \eqref{eqn:DefLocalQy}, it follows that
\begin{equation}
  e^{\i \theta M}\q(y)e^{-\i\theta M} = e^{\i \theta (\ell+2)/2}\q(e^{-\i \theta}y).
\end{equation}
In particular, for $y =0$ we have $e^{\i \theta M}\q e^{-\i\theta M} = e^{\i \theta (\ell+2)/2}\q$. Using \eqref{eqn:HamiltonianDensity}, we find that the Hamiltonian density $h$ of $\q$ conserves the magnetisation, $[h,M]=0$. Since $h(y)=h$ for any $y$ according to \cref{prop:HamDensity}, we conclude that the bulk part of the Hamiltonian $H(y)$ conserves the magnetisation.

For $y=0$, the boundary interactions are given by \eqref{eqn:hBND}, which obviously has the property $[h_{\text{\rm \tiny B}},M]=0$. This implies that for $y=0$, the full Hamiltonian conserves the magnetisation: $[H,M]=0$. Conversely, for generic values of $y$ this conservation law is broken by the boundary terms, which are given by
\begin{equation}
  h_{\text{\rm \tiny B}}(y) = \frac{1}{2}\q(y)^\dagger \q(y).
  \label{eqn:BoundaryTerm}
\end{equation}
We have $e^{\i \theta M}h_{\text{\rm \tiny B}}(y)e^{-\i\theta M} = h_{\text{\rm \tiny B}}(e^{-\i \theta}y)$.
It is possible albeit tedious to explicitly compute the matrix elements of this operator with respect to the canonical basis for arbitrary $\ell$. Their evaluation shows that the dependence on $y$ is non-trivial for both diagonal and off-diagonal matrix elements.

\paragraph{Parity and spin reversal.} The Hamiltonian density $h(y)=h$ is parity-invariant. Indeed, we have $PhP=h$ on $V^2$ as a consequence of the parity-invariance of the local supercharge $\q$. Since the boundary interactions are the same at both ends of the chain, we conclude that for any value of $y$ the Hamiltonian of the spin chain is parity-invariant:
\begin{equation}
  PH(y)P=H(y).
\end{equation}
The transformation property \eqref{eqn:ROnQy} of the local supercharge $\q(y)$ under spin reversal imply that the Hamiltonian density $h(y)=h$ is spin-reversal invariant: We have $RhR=h$ on $V^2$. Because of the boundary terms, this spin-reversal invariance does however not extend to the Hamiltonian of the spin chain for generic values of $y$. We have
\begin{equation}
  RH(y)R = H(y^{-1})
\end{equation}
and therefore spin-reversal invariance if and only if $y = \pm 1$.
 
\paragraph{Example.} As an example, we discuss the case $\ell=1$. The action of $\q(y)$ on the basis vectors $|0\rangle,|1\rangle$ is given by\footnote{
Here and in the following, we omit the commas whenever the sequences that label the basis vectors take specific values. For example, we write $\ket{01}=\ket{0,1}=\ket{0}\otimes \ket{1}$.}
\begin{align}
  \q(y)|0\rangle &= x\left(-2y|00\rangle+y^3|11\rangle -y^2(|01\rangle+|10\rangle)\right),\\
  \q(y)|1\rangle &= x\left(|00\rangle-2y^2|11\rangle -y(|01\rangle+|10\rangle)\right).
\end{align}
One checks the independence of the Hamiltonian density on $y$ by an explicit calculation. In order to write the boundary terms $h_{\text{\rm \tiny B}}(y)$ in a convenient way, we set $y=\rho e^{\i\theta}, \, \rho = |y|,$ and express $h_{\text{\rm \tiny B}}(y)$ in terms of the identity matrix and the Pauli matrices. We find
\begin{subequations}
\label{eqn:BoundaryHamXXZ}
  \begin{equation}
  h_{\text{\rm \tiny B}}(y) =  \left(\frac{1+5\rho^2+\rho^4}{4(1-\rho^2+\rho^4)}\right)\bm 1+\sum_{j=1}^3 \lambda_j\sigma^j,
\end{equation}
with 
\begin{equation}
   \lambda_1 = -\frac{\rho \cos \theta}{1+\rho^2},\quad  \lambda_2 = -\frac{\rho \sin \theta}{1+\rho^2},\quad \lambda_3 = -\frac{1}{4}\left(\frac{1-\rho^2}{1+\rho^2}\right).
\end{equation}
\end{subequations}
We conclude that the open XXZ chain at $\Delta=-1/2$ is supersymmetric for the family of boundary interactions \eqref{eqn:BoundaryHamXXZ}, parametrised by $\rho \geqslant 0$ and $\theta$. For $\rho \neq 0,\infty$ the off-diagonal terms of $h_{\text{\rm \tiny B}}$ are non-zero and therefore generalise the diagonal boundary interactions of \eqref{eqn:XXZOpenChainND} found by Fendley and Yang \cite{yang:04}.

\subsection{Boundary terms}
\label{sec:BoundaryConditions}

In this section, we modify the action of $\Q(y)$ defined in \eqref{eqn:DefQy} on the first and last site of the spin chain. This allows us to show that the lattice supersymmetry can be present for  unequal boundary terms at both ends of the spin chain. Both these boundary terms depend on the parameter $y$. Furthermore, each boundary term is individually characterised by an integer label $j=0,\dots,\ell+1$.

The main ingredient of our construction are the vectors
\begin{equation}
  |\xi_k(y)\rangle = x(|\phi(y)\rangle-|\phi(q^{2(k+1)}y)\rangle), \quad k=0,\dots, \ell+1,
  \label{eqn:DefXi}
\end{equation}
where $x$ is defined in \eqref{eqn:DefLocalQy} and $|\phi(y)\rangle$ in \eqref{eqn:DefPhiGauge}. Because of $q^{2(\ell+2)}=1$, we trivially have $|\xi_{\ell+1}(y)\rangle=0$. The action of $\q(y)$ on these vectors is very simple. Indeed, using \eqref{eqn:EqnForPhi} it is not difficult to show that
\begin{equation}
  \q(y)|\xi_k(y)\rangle = |\xi_k(y)\rangle\otimes |\xi_k(y)\rangle,\quad \, k=0,\dots, \ell+1.
  \label{eqn:ActionQXi}
\end{equation}

For any pair $0\leqslant j,k \leqslant \ell+1$, we consider an operator $\Q_{j,k}(y)$ that acts on any $|\psi\rangle \in V^L$ according to
\begin{equation}
  \label{eqn:DefQjk}
  \Q_{j,k}(y)|\psi\rangle = |\xi_j(y)\rangle \otimes |\psi\rangle + (-1)^{L-1} |\psi\rangle \otimes |\xi_k(y)\rangle + \Q(y)|\psi\rangle .
\end{equation}
We note that $\Q_{\ell+1,\ell+1}(y)=\Q(y)$. Using $\Q(y)^2=0$ and \eqref{eqn:ActionQXi}, one checks that
\begin{equation}
  \Q_{j,k}(y)^2=0.
\end{equation}

The corresponding Hamiltonian $H_{j,k}(y) = \Q_{j,k}(y)\Q_{j,k}(y)^\dagger + \Q_{j,k}(y)^\dagger \Q_{j,k}(y)$ is readily evaluated. It is given by a sum of nearest-neighbour interactions and boundary terms that depend on $j$ and $k$:
\begin{equation}
  H_{j,k}(y) = \sum_{i=1}^{L-1} h_{i,i+1} + (h_{\text{\rm \tiny B}}^{(j)}(y))_1+ (h_{\text{\rm \tiny B}}^{(k)}(y))_L.
\end{equation}
Here $h$ is the Hamiltonian density of $\q(y)$ and the boundary terms are given by
\begin{equation}
  h_{\text{\rm \tiny B}}^{(k)}(y) = h_{\text{\rm \tiny B}}(q^{2(k+1)}y). 
\end{equation}

We conclude that all boundary conditions that result from the deformation \eqref{eqn:DefLocalQy} of the local supercharge and a modification of the action of the supercharge on the first and last site of the spin chain are parameterised by a complex number $y$ and two integers $0\leqslant j,k\leqslant \ell+1$. However, not all choices of these parameters lead to unequal spectra.
To see this, we note that the spectrum of $H_{j,k}(y)$ is the same as the spectrum of $e^{\i \theta M} H_{j,k}(y)e^{-\i \theta M} = H_{j,k}(e^{-\i \theta}y)$ (for any real value of $\theta$) and $(H_{j,k}(y))^\ast = H_{J,K}(y^\ast)$ where $J=\ell-j \,(\text{mod}\, \ell+2),\,K=\ell-k \,(\text{mod}\, \ell+2)$. An appropriate choice for $\theta$ allows us to conclude that it is sufficient to 
restrict the parameters to real values for $y$, $j=\lfloor (\ell+1)/2\rfloor , \dots,\ell+1$ and $k=\ell+1$. For example, if $\ell=1$ there are two distinct cases $j=1,k=2$ and $j=k=2$ for any real value of $y$.

\section{Zero-energy states}

\label{sec:E0States}
In this section, we analyse whether the family of Hamiltonians $H_{j,k}(y)$ possesses so-called supersymmetry singlets or zero-energy states. If they exist then the zero-energy states are the ground states of the Hamiltonian.

In \cref{sec:Cohomology}, we recall a few basic facts about the spectrum of supersymmetric Hamiltonians. Moreover, we explain the relation between zero-energy states and the so-called cohomology of the supercharge. As we shall see, the structure of the cohomology depends on whether the parameter $y$ is non-zero or zero. We compute the cohomology for $y\neq 0$ in \cref{sec:GenericY} and for $y=0$ in \cref{sec:Y0}.

\subsection{Spectrum, zero-energy states and cohomology}

\label{sec:Cohomology}

Let us recall the characteristics of the spectrum and the eigenstates of a supersymmetric Hamiltonian $H=\Q\Q^\dagger + \Q^\dagger \Q,\,\Q^2=0$ for a generic supercharge $\Q$ \cite{witten:82}. Clearly, $H$ is a Hermitian operator and therefore diagonalisable. Furthermore, its spectrum is non-negative. Indeed, the Schr\"odinger equation $H|\psi\rangle = E |\psi\rangle$ implies
\begin{equation}
  ||\Q|\psi\rangle||^2+||\Q^\dagger|\psi\rangle||^2 = E ||\psi||^2.
  \label{eqn:SchroedingerProjected}
\end{equation}
Since the left-hand side is non-negative and the norm of an eigenvector $|\psi\rangle$ non-vanishing, we must have $E\geqslant 0$. We call the non-zero solutions of the Schr\"odinger equation with $E>0$ and $E=0$ positive-energy states and zero-energy states, respectively. They differ in their behaviour under the action of the supercharge and its adjoint.

\paragraph{Positive-energy states.} The eigenstates with strictly positive energy $E>0$ organise in \textit{doublets}. They
are given by a pair of non-zero vectors $|\psi\rangle, \,\Q|\psi\rangle$ with $\Q^\dagger|\psi\rangle = 0$.
The two states in the doublet are called superpartners. Since the Hamiltonian commutes with the supercharge, the superpartners have the same eigenvalue $E$.

In our case, the supercharge increases the length of the chain. Hence, the supersymmetry leads to spectral degeneracies for chains with lengths differing by one. \cref{fig:SpectraXXZ} illustrates this spectral degeneracy for the Hamiltonians $H_{j,k}(y)$ with $\ell=1$ and unequal boundary conditions at both ends of the chain.
\begin{figure}[h]
  \centering
  \includegraphics[width=.475\textwidth]{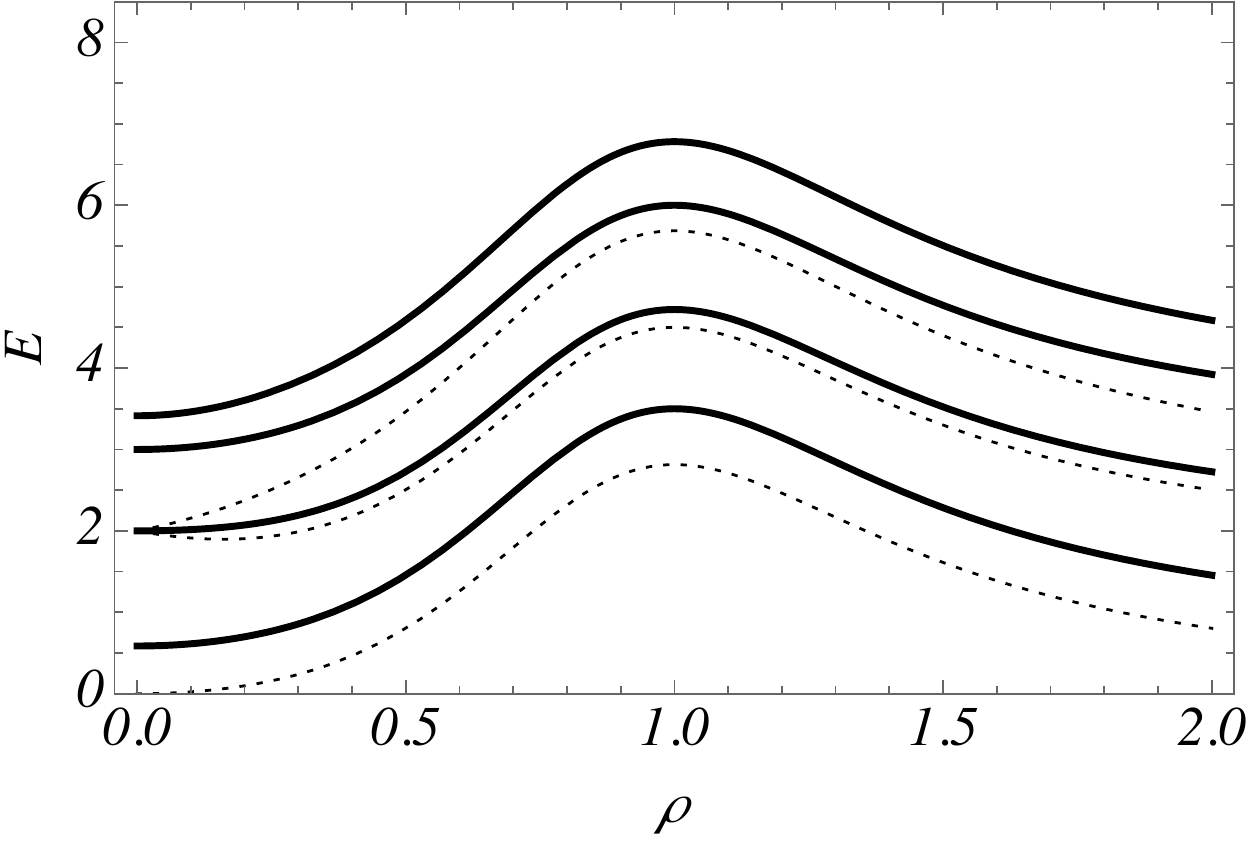}
  \hspace{3mm}
  \includegraphics[width=.475\textwidth]{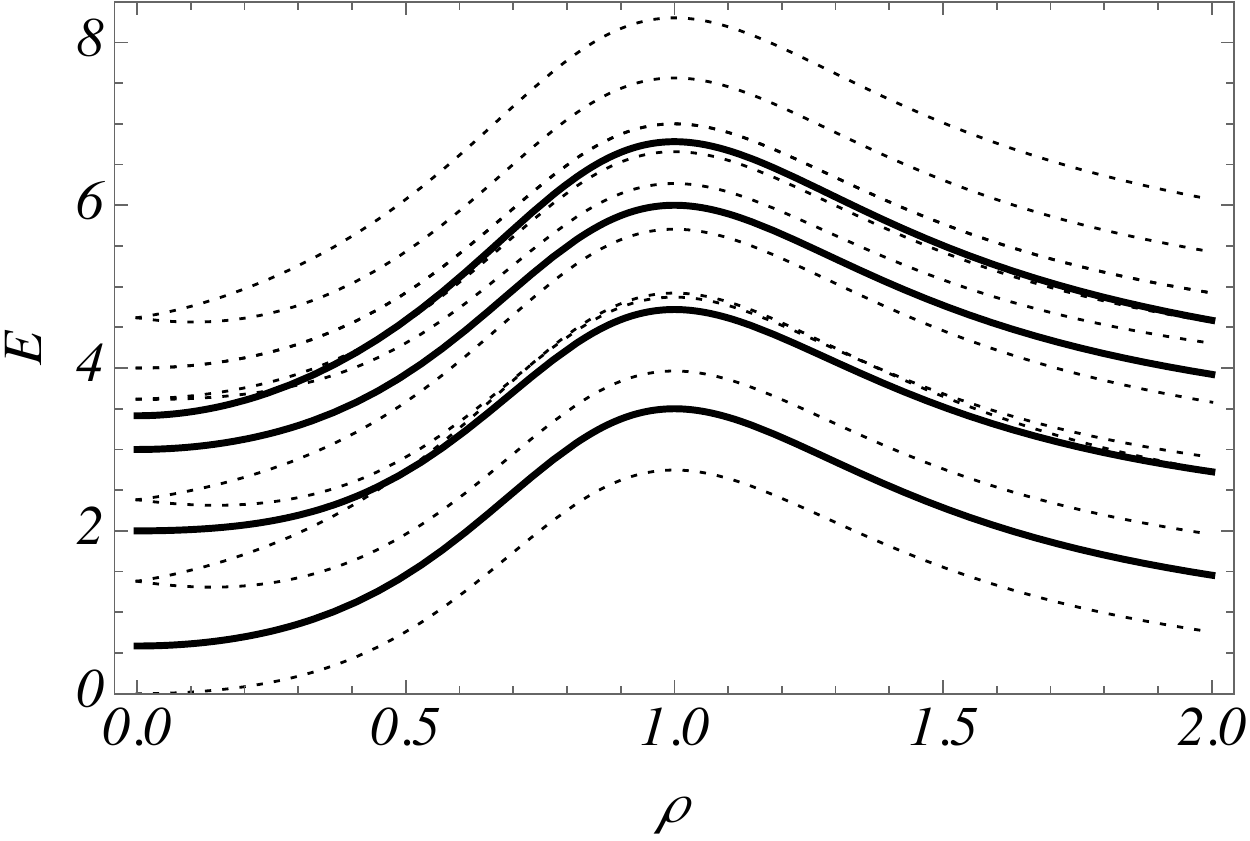}
  \caption{The spectrum of the Hamiltonian $H_{j,k}(y)$ for $\ell=1$ with boundary conditions labelled by $j=1,k=2$ for $L=3$ (left panel) and $L=4$ sites (right panel) as a function of $\rho=|y|$. The solid lines correspond to exact common eigenvalues in the two spectra.} 
  \label{fig:SpectraXXZ}
\end{figure}

\paragraph{Zero-energy states.} In the following, we focus on the solutions of the Schr\"odinger equation with $E=0$. Since the spectrum of the Hamiltonian is non-negative, any non-zero solution of $H|\psi\rangle = 0$ is automatically a ground state of the system. The existence of these zero-energy states is, according to  \eqref{eqn:SchroedingerProjected}, equivalent to the existence of non-zero solutions of the system of equations
\begin{equation}
  \label{eqn:ZeroEnergyStateEqns}
  \Q|\psi\rangle = 0, \quad \Q^\dagger |\psi\rangle = 0.
\end{equation}
These equations imply that the zero-energy states are \textit{singlets} (as stated in the introduction) in the sense that no other eigenstates of the Hamiltonian can be obtained by acting on them with the supercharge or its adjoint.

The first equation of \eqref{eqn:ZeroEnergyStateEqns} requires that a zero-energy state be in the kernel of the supercharge. We call the elements of $\ker \Q$ \textit{cocycles}. Since $\Q^2=0$, the kernel contains all states that are in the image of $\Q$. We call the elements of $\text{im}\, \Q$ \textit{coboundaries}. The second equation of \eqref{eqn:ZeroEnergyStateEqns} leads to the following property of zero-energy states:

\begin{lemma}
  \label{lem:E0NotCoboundary}
  A zero-energy state is not a coboundary.
  \begin{proof}
    Let $|\psi\rangle$ be a zero-energy state and assume that $|\psi\rangle= \Q|\phi\rangle$ for some vector $|\phi\rangle$. Then its square norm is $\langle \psi|\psi\rangle = \langle \psi|(\Q|\phi\rangle) = \left(\langle \phi|(\Q^\dagger |\psi\rangle)\right)^\ast = 0$ where we used \eqref{eqn:ZeroEnergyStateEqns}. Hence, $|\psi\rangle=0$, which is a contradiction.
  \end{proof}
\end{lemma}

\paragraph{Cohomology.} \cref{lem:E0NotCoboundary} suggests that we consider the kernel of the supercharge modulo its image in order to analyse the existence and the properties of the zero-energy states. We define therefore the \textit{cohomology} of the supercharge
\begin{equation}
  \mathcal H^\bullet(\Q) = \bigoplus_{L=1}^\infty \mathcal H^L(\Q),
\end{equation}
where $\mathcal H^1(\Q) = \text{ker}\{ \Q:V^1\to V^{2}\}$ and
\begin{equation}
   \mathcal H^L(\Q) = \frac{\text{ker}\{ \Q:V^L\to V^{L+1}\}}{\text{im}\{ \Q:V^{L-1}\to V^{L}\}}, \quad \text{for}\quad L\geqslant 2.
\end{equation}
The elements of $\mathcal H^L(\Q)$ are equivalence classes or cohomology classes. Any cohomology class of $\mathcal H^L(\Q)$ can be represented by a cocycle $|\psi\rangle\in V^L$, which is called \textit{its} \textit{representative}. Any two cocycles of $V^L$ differing by a coboundary represent the same element of $\mathcal H^L(\Q)$. We denote the equivalence class of a cocycle $|\psi\rangle \in V^L$ by $[|\psi\rangle]$. Hence $[|\psi\rangle + \Q|\phi\rangle] = [|\psi\rangle]$ for all $|\phi\rangle \in V^{L-1}$. If $\mathcal H^L(\Q)=0$ for each $L\geqslant 1$ then we call the cohomology trivial. This is the case if and only if all cocycles are coboundaries. 

It can be shown \cite{witten:82} that $\mathcal H^L(\Q)$ is isomorphic to the subspace of $V^L$ that is spanned by the zero-energy states of the Hamiltonian $H$ for a chain of length $L$. Furthermore, if $|\phi\rangle \in V^L$ is a representative of a non-zero element of $\mathcal H^L(\Q)$, then there is a state $|\phi'\rangle \in V^{L-1}$ such that
\begin{equation}
  |\psi\rangle = |\phi\rangle + \Q|\phi'\rangle
  \label{eqn:RepresentativeE0State}
\end{equation}
is a zero-energy state \cite{witten:82}. In the following, we use this connection with cohomology classes and their representatives to investigate some properties of the zero-energy states of our models.

\subsection{The case \texorpdfstring{$y\neq 0$}{YNot0}}

\label{sec:GenericY}
In this and the following subsection, we explicitly compute $\mathcal H^L(\Q_{j,k}(y))$. This computation allows us to characterise the space of zero-energy states of the Hamiltonian $H_{j,k}(y)$ as a function of the parameter $y$, the integer labels $j,k$ and the system size $L$.

Here, we consider the case $y\neq 0$. We prove the following theorem:
\begin{theorem}
  \label{thm:CohomGenericY}
  For $y\neq 0$ and each $j,k=0,\dots,\ell+1$, the cohomology $\mathcal{H}^\bullet (\Q_{j,k}(y))$ is trivial.
\end{theorem}

This theorem implies that for $y\neq 0$ and any length of the chain $L$ the Hamiltonian $H_{j,k}(y)$ does not possess zero-energy states. Thus its spectrum is strictly positive.

The proof is based on two lemmas. The first lemma deals with a mapping $s$ that is akin to a so-called \textit{contracting homotopy} \cite{loday:92}.

\begin{lemma}
\label{lem:ContractingHomotopy}
  Let $\Q$ be an arbitrary supercharge. Suppose that for each $L\geqslant 2$ there is a mapping $s:V^L\to V^{L-1}$ such that
  \begin{equation}
s\Q+\Q s = 1.
\label{eqn:ContractingHomotopy}
\end{equation}
  Then for each $L\geqslant 2$, we have $\mathcal H^L(\Q) = 0.$
  \begin{proof}
    We show that any cocycle $|\psi\rangle \in V^L$ is a coboundary.  Indeed, applying \eqref{eqn:ContractingHomotopy} to $|\psi\rangle$ we obtain
    \begin{equation}
      |\psi\rangle = (s\Q + \Q s)|\psi\rangle = \Q(s|\psi\rangle).
    \end{equation}
    Hence, $\mathcal H^L(\Q) =0$.
  \end{proof}
\end{lemma}

Our aim is to construct such a mapping $s$ for the supercharge $\Q_{j,k}(y)$. To this end, we use the vectors $|\xi_0(y)\rangle,\dots,|\xi_\ell(y)\rangle$, defined in \eqref{eqn:DefXi}. The second lemma needed for our proof of \cref{thm:CohomGenericY} establishes that for non-vanishing $y$ these vectors span the Hilbert space $V$ of a single spin:

\begin{lemma}
  \label{lem:BasisVectors}
   For $y\neq 0$ the vectors $|\xi_0(y)\rangle,\dots,|\xi_\ell(y)\rangle$ constitute a basis of $V$.
   \begin{proof}
     The components of $|\xi_n(y)\rangle$ with respect to the canonical basis are given by
     \begin{equation}
       \Xi_{mn} = \langle m |\xi_n(y)\rangle = x \frac{y^{m+1}}{\sqrt{\{m+1\}}}(1-q^{2(m+1)(n+1)}).
     \end{equation}
 To prove the lemma, it is sufficient to show that the matrix $\Xi = (\Xi_{mn})_{m,n=0}^\ell$ is invertible. This is indeed the case. One checks that the entries of the inverse matrix are given by 
      \begin{equation}
    \left(\Xi^{-1}\right)_{mn} = - \frac{\sqrt{\{n+1\}}}{(\ell+2)xy^{n+1}}q^{-2(m+1)(n+1)}.
  \end{equation}
This ends the proof.
   \end{proof}
\end{lemma}

\begin{proof}[Proof of \cref{thm:CohomGenericY}] We now prove that if $y\neq 0$ then $\mathcal H^L(\Q_{j,k}(y))=0$ 
for each $L\geqslant 1$ and each $j,k=0,\dots,\ell+1$. For $L=1$, the proof is trivial: one readily checks that $\ker\{\Q_{j,k}(y):V\to V^2\}=0$, using \cref{lem:BasisVectors}.
Hence, we focus on $L\geqslant 2$. The proof is based on the construction of
a mapping $s_j$ that obeys \eqref{eqn:ContractingHomotopy} for each $j=0,\dots,\ell+1$. We separately consider the cases $0\leqslant j\leqslant \ell$ and $j=\ell+1$.

Let us first consider $0\leqslant j\leqslant \ell$. It follows from \cref{lem:BasisVectors} and standard properties of the tensor product that for $y\neq 0$, every vector $|\psi\rangle\in V^L$ can be written as
\begin{equation}
  |\psi\rangle = \sum_{m=0}^\ell |\xi_m(y)\rangle \otimes |\psi_m\rangle
\end{equation}
with unique vectors $|\psi_0\rangle,\dots,|\psi_\ell\rangle \in V^{L-1}$. We define the mapping $s_j$ by
\begin{equation}
  s_j |\psi\rangle = |\psi_j\rangle.
  \label{eqn:DefSj}
\end{equation}
Using action \eqref{eqn:ActionQXi} of the local supercharge on the special basis vectors, it is easy to see that $(s_j \Q(y) + \Q(y) s_j)|\psi\rangle = -|\xi_j(y)\rangle\otimes|\psi_j\rangle = -|\xi_j(y)\rangle\otimes s_j|\psi\rangle$. We combine this identity with the definition of the supercharge \eqref{eqn:DefQjk} and find that
\begin{equation}
  \label{eqn:CommQs}
  s_j \Q_{j,k}(y) + \Q_{j,k}(y) s_j= 1
\end{equation}
for each $j=0,\dots,\ell$ and $k=0,\dots,\ell+1$.

Second, for $j=\ell+1$, we define
\begin{equation}
  s_{\ell+1}=-\sum_{j=0}^\ell s_j.
\end{equation}
Using the definition of $s_j$, it is easy to see that \eqref{eqn:CommQs} holds for $j=\ell+1$ and $k=0,\dots,\ell+1$, too.

In both cases, it follows from \cref{lem:ContractingHomotopy} $\mathcal H^L(\Q_{j,k}(y))=0$ for any $L\geqslant 2$ and each $j,k=0,\dots,\ell+1$. This ends the proof of the theorem.
\end{proof}
We notice that the proof only relies on the existence of a basis $|\xi_0(y)\rangle,\dots,|\xi_\ell(y)\rangle$ of $V$ with the property $\q(y)|\xi_k(y)\rangle=|\xi_k(y)\rangle\otimes |\xi_k(y)\rangle$ for each $k=0,\dots,\ell$. This is the case for a variety of other physically-relevant spin chains. An example is the quantum spin $1/2$ XYZ chain along a special line of couplings \cite{hagendorf:13}.

\subsection{The case \texorpdfstring{$y=0$}{YIs0}}
\label{sec:Y0}

The proof of \cref{thm:CohomGenericY} is not generalisable to the cases $y=0$ and $y \to \infty$. They have to be separately treated. In this subsection, we consider the case $y = 0$. For this value of $y$, the local supercharge reduces to $\q(y=0) = \q$, defined in \eqref{eqn:TrigSupercharge}. The supercharge $\Q_{j,k}(y=0)$ is independent of the indices $j,k$ and we simply denote it by $\Q$. Its cohomology is non-trivial. The results presented here below can easily be modified in order to cover the case $y\to \infty$. In this case, the local supercharge reduces to $\q(y\to \infty)=\bar \q$, which is the image of $\q$ under spin reversal. In fact, the representatives of $\mathcal H^L(\bar \Q)$ are simply obtained by applying the spin-reversal operator to the representatives of $\mathcal H^L(\Q)$.

The main result of this subsection is the following theorem:
\begin{theorem}
  \label{thm:OCCohom}
  $\mathcal H^L(\mathfrak Q)$ is spanned by the cohomology class of the state\footnote{We abbreviate the tensor product $|\chit\rangle \otimes \cdots \otimes |\chit\rangle$ by $|\chit \cdots \chit\rangle$ in order to simplify the notation.}
    \begin{subequations}
  \label{eqn:RepE0States}
  \begin{equation}
    |\underset{n\,\text{times}}{\underbrace{\chit\cdots\chit }}\rangle\quad \text{if}\quad L=2n,
  \end{equation}
and
   \begin{equation}
    |0\rangle\otimes |\hspace{-.12cm}\underset{n-1\,\text{times}}{\underbrace{\chit\cdots\chit }}\hspace{-.12cm}\rangle  \quad \text{if}\quad L=2n-1,
  \end{equation}
  \end{subequations}
  where $|\chit\rangle\in V^2$ is the state defined in \eqref{eqn:DefChi} and $n$ a positive integer.
\end{theorem}

This theorem implies that for $y=0$ the spin-chain Hamiltonian possesses a zero-energy state for each length $L$. The state is unique up to normalisation. This result is expected from the above-mentioned mapping between the spin-chain models at $y=0$ and the $M_\ell$ models of supersymmetric fermions on open intervals. For these fermion models, the dimension of the space of zero-energy states was indeed shown to be one-dimensional in previous works \cite{huijse:10_1,huijse:15}. To our knowledge, the structure of the corresponding cohomology has however so far remained undetermined and unexploited. \cref{thm:OCCohom} provides this structure for the spin-chain models. We use it in the forthcoming sections in order to compute certain non-trivial scalar products that involve the zero-energy states.

The proof of \cref{thm:OCCohom} is based on several lemmas. They establish the existence of an explicit bijection between
$\mathcal H^{L}(\Q)$ and $\mathcal H^{L+2}(\Q)$ for each $L\geqslant 1$.
Hence, we may construct $\mathcal H^{L}(\Q)$ from $\mathcal H^{1}(\Q)$ and $\mathcal H^{2}(\Q)$. We explicitly compute them 
in the following lemma:

\begin{lemma}
\label{lem:lowdim}
$\mathcal H^1(\Q)$ and $\mathcal H^2(\Q)$ are spanned by the cohomology classes of the states $|0\rangle$ and $|\chit\rangle$, respectively.
\begin{proof}
For $L=1$, recall that $\mathcal H^1(\Q) =\ker\{\Q:V^1\to V^2\}$.
  According to \eqref{eqn:TrigSupercharge}, the only solution to $\q|\psi\rangle = 0$ is $|\psi\rangle = |0\rangle$, up to a factor. 

For $L=2$, we consider a cocycle $|\psi\rangle \in V^2$. We write $\psi_{m,n} = \langle m,n|\psi\rangle$ for its components with respect to the canonical basis of the Hilbert space. From $\Q|\psi\rangle = 0$, it follows that
\begin{equation}
  \psi_{m,\ell}=\psi_{\ell,m} =0 \quad \text{for } m = 1,\dots,\ell
  \label{eqn:APsi1}
\end{equation}
and
\begin{equation}
  a_{m,0}\psi_{m,n} = a_{n+1,0} \psi_{m-1,n+1}\quad \text{for } m = 1,\dots,\ell\text{ and } n=0,\dots,\ell-1.
  \label{eqn:APsi2}
\end{equation}
We distinguish three cases. \textit{(i)} For $m+n > \ell$, we find from \eqref{eqn:APsi2} the relation
\begin{equation}
  \psi_{m,n} = \psi_{m+n-\ell,\ell}\prod_{j=n+1}^\ell \frac{a_{j,0}}{a_{m+n+1-j,0}}.
  \label{eqn:PsiRec}
\end{equation}
From \eqref{eqn:APsi1}, it follows that $\psi_{m,n}=0$. \textit{(ii)} For $m+n=\ell$, the relation \eqref{eqn:PsiRec} still holds. The explicit form of the coefficients $a_{m,k}$, defined in \eqref{eqn:TrigSupercharge}, leads to
\begin{equation}
  \psi_{m,\ell-m} = \psi_{0,\ell} \chit_m
\end{equation}
where $\chit_m,\, m=0,\dots,\ell,$ are the components of the vector $|\chit\rangle$ defined in \eqref{eqn:DefChi}.
\textit{(iii)} For $m+n< \ell$, we obtain from \eqref{eqn:APsi2} the relation
\begin{equation}
\psi_{m,n} = \psi_{0,m+n}\prod_{j=1}^m \frac{a_{n+j,0}}{a_{j,0}} = \psi_{0,m+n} \frac{a_{m+n+1,m}}{a_{m+n+1,0}}.
  \label{eqn:PsiRec2}
\end{equation}

Combining \textit{(i)}, \textit{(ii)} and \textit{(iii)}, we find after some algebra
\begin{equation}
  |\psi\rangle = \psi_{0,\ell}|\chit\rangle + \Q\left(\sum_{p=0}^{\ell-1} \frac{\psi_{0,p}}{a_{p+1,0}}|p+1\rangle\right).
\end{equation}
Hence $[|\psi\rangle]=\psi_{0,\ell}[|\chit\rangle]$ with an arbitrary coefficient $\psi_{0,\ell}$. The cohomology class of $|\chit\rangle$ cannot be zero since this state is a linear combination of basis states of $V^2$ which are clearly not in the image of $\q$.
  \end{proof}
\end{lemma}

Our next aim is to study $\mathcal H^L(\mathfrak Q)$ for $L\geqslant 3$. In the following technical lemma, we determine a convenient choice of their representatives.

\begin{lemma}
  \label{lem:technical}
  For each $L\geqslant3$ any element in $\mathcal H^L(\mathfrak Q)$ can be represented by a cocycle $|\psi\rangle \in V^L$ with $|\psi\rangle = \sum_{m=0}^\ell |\psi_m\rangle \otimes |m\rangle$ such that $|\psi_0\rangle = |\psi'_{\ell,0}\rangle \otimes |\ell\rangle$ for some vector $|\psi'_{\ell,0}\rangle\in V^{L-2}$.
  \begin{proof}
 Let us a consider the cocycle $|\psi'\rangle$ representing an element of $\mathcal H^L(\Q)$. Then for any $|\phi\rangle\in V^{L-1}$ the vector $|\psi\rangle = |\psi'\rangle+\Q|\phi\rangle$ is also a cocycle, representing the same element of $\mathcal H^L(\Q)$. We write the vector $|\psi\rangle$ (and likewise $|\psi'\rangle$, $|\phi\rangle$) as a superposition
  \begin{equation}
        |\psi\rangle = \sum_{m=0}^\ell |\psi_m\rangle \otimes |m\rangle, \quad |\psi_m\rangle \in V^{L-1}.
        \label{eqn:decomposition}
  \end{equation}
  This leads to
 \begin{equation}
    |\psi_m\rangle = |\psi'_m\rangle+\Q|\phi_m\rangle+(-1)^L \sum_{k=m+1}^\ell a_{k,m}|\phi_k\rangle \otimes |k-m-1\rangle.
  \end{equation}
 In order to prove the lemma, we consider the case $m=0$:
 \begin{equation}
   \label{eqn:psi0}
    |\psi_0\rangle = |\psi'_0\rangle+\Q|\phi_0\rangle+(-1)^L \sum_{k=1}^\ell a_{k,0}|\phi_k\rangle \otimes |k-1\rangle.
  \end{equation}
  Again, we decompose the vector $|\psi'_0\rangle$ with respect to the last site, $|\psi'_0\rangle = \sum_{m=0}^\ell|\psi'_{m,0}\rangle \otimes |m\rangle$, and substitute this decomposition into \eqref{eqn:psi0}. The choices
  \begin{equation}
    |\phi_0\rangle = 0, \quad \text{and} \quad |\phi_k\rangle = (-1)^{L+1}a_{k,0}^{-1} |\psi'_{k-1,0}\rangle, \quad \text{for }k=1,\dots,\ell
  \end{equation}
  lead to $|\psi_0\rangle = |\psi_{\ell,0}'\rangle\otimes |\ell\rangle$. This ends the proof.
\end{proof}
\end{lemma}

For the next two lemmas, we introduce the operator $S$ which acts on any vector $|\psi\rangle \in V^L,\,L\geqslant 1,$ according to
\begin{equation}
  S|\psi\rangle = |\psi\rangle \otimes |\chit\rangle.
\end{equation}
One checks that it commutes with the supercharge $\Q$,
\begin{equation}
  S \Q = \Q S
\end{equation}
because $|\chit\rangle$ is annihilated by the supercharge. It follows that $S$ can be extended to a mapping $S^\sharp$ defined on the cohomology \cite{loday:92,masson:08}: Its action on the cohomology classes is given by $S^\sharp[|\psi\rangle] = [S|\psi\rangle]$.

\begin{lemma}
  \label{lem:surjectivity}
  For each $L\geqslant 1$ the mapping $S^\sharp:\mathcal H^{L}(\Q)\to \mathcal H^{L+2}(\Q)$ is surjective.
  \begin{proof}
  Let $|\psi\rangle$ be a cocycle representing an element of $\mathcal H^{L+2}(\Q)$.   We decompose it as in \eqref{eqn:decomposition} with respect to the last site. The equation $\Q|\psi\rangle = 0$ leads to
  \begin{equation}
  \Q|\psi_m\rangle = (-1)^{L+1}\sum_{k=m+1}^\ell a_{k,m}|\psi_k\rangle \otimes |k-(m+1)\rangle.
  \end{equation}
  Let us consider the case $m=0$. From \cref{lem:technical} we know that without loss of generality $|\psi\rangle$ can be chosen in such a way that $|\psi_0\rangle = |\psi_{\ell,0}\rangle\otimes |\ell\rangle$ for some state $|\psi_{\ell,0}\rangle\in V^{L}$. This choice leads to
  \begin{equation}
    \Q|\psi_{\ell,0}\rangle \otimes |\ell\rangle +(-1)^{L+1} |\psi_{\ell,0}\rangle \otimes \sum_{k=0}^{\ell-1}a_{\ell, k}|k,\ell-k-1\rangle=(-1)^{L+1}\sum_{k=1}^\ell a_{k,0}|\psi_k\rangle \otimes |k-1\rangle. 
  \end{equation}
  Comparing both sides, we obtain 
  \begin{equation}
    \Q|\psi_{\ell,0}\rangle =0,\quad\text{and}\quad  |\psi_k\rangle = \left(\frac{a_{\ell,\ell-k}}{a_{k,0}}\right)|\psi_{\ell,0}\rangle \otimes |\ell-k\rangle, \quad \text{for } k=1,\dots,\ell.
  \end{equation}
  According to \eqref{eqn:TrigSupercharge} and \eqref{eqn:DefChi}, $a_{\ell,\ell-k}/a_{k,0}=\chit_k$. Hence, we find
  \begin{equation}
    |\psi\rangle = |\psi_{\ell,0}\rangle\otimes |\ell, 0\rangle + |\psi_{\ell,0}\rangle \otimes \sum_{k=1}^\ell \chit_k|\ell-k,k\rangle = S|\psi_{\ell,0}\rangle
  \end{equation}
  with a cocycle $|\psi_{\ell,0}\rangle \in V^{L}$. For the corresponding cohomology classes we find thus $[|\psi\rangle]=S^\sharp[|\psi_{\ell,0}\rangle]$.
\end{proof}
\end{lemma}

\begin{lemma}
  \label{lem:injectivity}
  For each $L\geqslant 1$ the mapping $S^\sharp:\mathcal H^{L}(\Q)\to \mathcal H^{L+2}(\Q)$ is injective.
  \begin{proof}
  Consider an element of $\ker\{ S^\sharp:\mathcal H^{L}(\Q)\to\mathcal H^{L+2}(\Q)\}$. It can be represented by a cocycle $|\psi\rangle \in V^L$ such that $S|\psi\rangle = \Q|\phi\rangle$ for some vector $|\phi\rangle\in V^{L+1}$. As before, it is useful to decompose the state with respect to the last site: $|\phi\rangle=\sum_{m=0}^\ell |\phi_m\rangle \otimes |m\rangle$. We find
  \begin{equation}
    S|\psi\rangle = |\psi\rangle \otimes |\chit\rangle = \sum_{m=0}^\ell \Q|\phi_m\rangle \otimes |m\rangle + (-1)^{L+1} \sum_{m=0}^\ell\sum_{k=0}^{m-1}a_{m,k}|\phi_m\rangle \otimes |k,m-k-1\rangle.
  \end{equation}
  We select on both sides the terms corresponding to $|0\rangle$ on the last site, and find
  \begin{equation}
    \chit_\ell|\psi\rangle \otimes |\ell\rangle = \Q|\phi_0\rangle  + (-1)^{L+1}\sum_{m=0}^\ell a_{m,m-1} |\phi_m\rangle \otimes |m-1\rangle.
    \label{eqn:intermed}
  \end{equation}
  Notice that the sum on the right-hand side does not contain any term proportional to $|\ell\rangle$ on the last site. $\Q|\phi_0\rangle$ may however contain such a term. To see this, we decompose $|\phi_0\rangle = \sum_{m=0}^\ell|\phi_{m,0}\rangle\otimes |m\rangle$. We apply the supercharge to this decomposition, insert it into \eqref{eqn:intermed} and obtain
  \begin{equation}
    \chit_\ell|\psi\rangle \otimes |\ell\rangle= \Q|\phi_{\ell,0}\rangle \otimes |\ell\rangle + \sum_{k=0}^{\ell-1}|\tilde \phi_k\rangle \otimes |k\rangle.
  \end{equation}
  The states $|\tilde \phi_k\rangle$ can in principle be computed but we won't need them. This equality implies $|\psi\rangle = \Q|\phi_{\ell,0}\rangle$. Hence, $|\psi\rangle$ is a coboundary. We conclude that $\ker \{S^\sharp :\mathcal H^{L}(\Q)\to\mathcal H^{L+2}(\Q)\}=0$. This proves the claim.
\end{proof}
\end{lemma}

We are now ready to prove the main result of this section.

\begin{proof}[Proof of \cref{thm:OCCohom}]
  From \cref{lem:surjectivity,lem:injectivity} we conclude for each $L\geqslant 1$ the mapping $S^\sharp: \mathcal H^{L}(\Q)\to \mathcal H^{L+2}(\Q)$ is both surjective and injective. Hence $ \mathcal H^{L+2}(\Q)$ is isomorphic to $\mathcal H^{L}(\Q)$. By transitivity, we obtain
  \begin{align}
    \mathcal H^{2n-1}(\Q) = (S^{\sharp})^{n-1}\mathcal H^{1}(\Q),\quad
   \mathcal H^{2n}(\Q) = (S^{\sharp})^{n-1}\mathcal H^{2}(\Q)
  \end{align}
for each $n\geqslant 1$. 
$\mathcal H^{1}(\Q)$ and $\mathcal H^{2}(\Q)$ were obtained in \cref{lem:lowdim}. This allows us to compute representatives of the elements of $\mathcal H^L(\Q)$ for odd and even $L$ from the repeated action of $S$ on $|0\rangle$ and $|\chit\rangle$, respectively, which leads to \eqref{eqn:RepE0States}. 
\end{proof}

\section{Components and scalar products of the zero-energy states}

\label{sec:SP}

In this section, we analyse the zero-energy states of the spin-chain Hamiltonians $H$ with $y=0$. Our main goal is to unveil some of their properties with the help of \cref{thm:OCCohom}. In \cref{sec:Representations}, we discuss two representations of the zero-energy states arising from the representatives of the cohomology and their homology analogues. We deduce from these representations their magnetisation, parity and relations between certain components. In \cref{sec:ScalarProducts}, we introduce a family of physically-interesting scalar products that involve an arbitrary number of zero-energy states for systems of different lengths. We use the supersymmetry to show that the knowledge of a single special component of each involved state is sufficient to evaluate the scalar product.

\subsection{Representation of zero-energy states}

\label{sec:Representations}
We denote by $|\psi_L\rangle$ a zero-energy state of the Hamiltonian $H$ for a chain of length $L$. For convenience, we allow $L=1$ even though $H$ is not defined for chains consisting of a single site. Furthermore, we write
\begin{equation}
  (\psi_L)_{m_1m_2\dots m_L} = \langle m_1m_2\cdots m_L|\psi_L\rangle
\end{equation}
for its components with respect to the canonical basis of $V^L$.

\paragraph{Cohomology representation.} It follows from \cref{thm:OCCohom} that $|\psi_L\rangle$ can be written as the sum of a multiple of a special representative of $\mathcal H^L(\Q)$ and an element of the image of the supercharge. Specifically, we have
\begin{proposition}
\label{prop:CohomRepr}
For each $n\geqslant 1$, we have
\begin{equation}
  |\psi_{2n}\rangle = \lambda_{2n}|\chit\cdots \chit\rangle +\Q|\phi_{2n}\rangle, \quad \lambda_{2n}=(\psi_{2n})_{0\ell\cdots 0\ell},
  \label{eqn:psi2n}
\end{equation}
with $|\phi_{2n}\rangle\in V^{2n-1}$ and
\begin{equation}
  |\psi_{2n-1}\rangle =\lambda_{2n-1}  |0\rangle\otimes |\chit\cdots \chit\rangle  +\Q|\phi_{2n-1}\rangle, \quad \lambda_{2n-1}= (\psi_{2n-1})_{0\ell\cdots 0\ell 0},
  \label{eqn:psi2n1}
\end{equation}
with  $|\phi_{2n-1}\rangle\in V^{2(n-1)}$.
\begin{proof} 
  For simplicity, we focus on the case $L=2n$. According to \eqref{eqn:RepresentativeE0State} and \cref{thm:OCCohom}, there is a scalar $\lambda_{2n}$ and a state $|\phi_{2n}\rangle\in V^{2n-1}$ such that the zero-energy state can be written as
  \begin{equation}
    |\psi_{2n}\rangle = \lambda_{2n} |\chit\cdots \chit\rangle +\Q|\phi_{2n}\rangle.
  \end{equation}
  The state $|\phi_{2n}\rangle$ cannot be obtained from $\mathcal H^{2n}(\Q)$. It is fixed by the requirement that the zero-energy state be annihilated by the adjoint supercharge.
  
  We notice however that the unknown term $\Q|\phi_{2n}\rangle$ does not contribute to the scalar product of the zero-energy state with any vector $|\omega\rangle \in \text{ker}\{\Q^\dagger:V^{2n}\to V^{2n-1}\}$. Indeed, we have
  \begin{align}
  \langle \omega |\psi_{2n}\rangle &= \lambda_{2n} \langle \omega |\chit\cdots \chit\rangle+ \langle \omega |\Q|\phi_{2n}\rangle\nonumber= \lambda_{2n} \langle \omega |\chit\cdots \chit\rangle+ \langle \phi_{2n} |\Q^\dagger|\omega\rangle^\ast\nonumber \\
  &=\lambda_{2n} \langle \omega |\chit\cdots \chit\rangle \label{eqn:SPPsi2n}.
\end{align}
A simple choice for $|\omega\rangle$ is given by a canonical basis vector $ |\omega\rangle = |m_1,\dots,m_{2n}\rangle \in 
V^{2n}$ that \textit{(i)} has a non-zero scalar product with $|\chit\cdots \chit\rangle$ and \textit{(ii)} is annihilated by the adjoint supercharge. The requirement \textit{(i)} leads to the constraint
\begin{subequations}
  \label{eqn:ConstraintsBasisVector}
\begin{equation}
  m_{2i-1}+m_{2i} = \ell, \quad i=1,\dots,n,
\end{equation}
because the magnetisation of the state $|\chit\rangle$ vanishes. The requirement \textit{(ii)} leads to the constraint
\begin{equation}
  m_{2i}+m_{2i+1} \geqslant \ell, \quad i=1,\dots,n-1,
\end{equation}
which follows from the action \eqref{eqn:ActionAdjointTrigSC} of the adjoint local supercharge $\q^\dagger$.
\end{subequations}

The constraints \eqref{eqn:ConstraintsBasisVector} have many solutions. One solution is given by
\begin{equation}
  |\omega \rangle =|0\ell \cdots 0\ell\rangle.
\end{equation}
Using \eqref{eqn:DefChi}, we find that $\langle \omega |\chit\cdots \chit\rangle = 1$ and thus $\lambda_{2n}= \langle \omega |\psi_{2n}\rangle = (\psi_{2n})_{0\ell\cdots 0\ell}$. This concludes the proof for $L=2n$.

The proof for $L=2n-1$ is similar and uses the
vector $|\omega \rangle = |0\ell\cdots 0\ell 0\rangle$.
\end{proof}
\end{proposition}

The representations of the zero-energy states given in \eqref{eqn:psi2n} and \eqref{eqn:psi2n1} allow us to derive a number of simple properties of the zero energy states.
Two immediate consequences are:
\begin{corollary}
  \label{corr:NonZeroComps}
  For each $n\geqslant 1$ the components $(\psi_{2n})_{0\ell\cdots 0\ell}$ and $(\psi_{2n-1})_{0\ell\cdots 0\ell 0}$ are non-zero.
  \begin{proof}
    Consider $L=2n$ and suppose that $(\psi_{2n})_{0\ell\cdots 0\ell}=0$. Then it follows from \eqref{eqn:psi2n} that $|\psi_{2n}\rangle \in \text{im}\,\Q$ which contradicts \cref{lem:E0NotCoboundary}. The case $L=2n-1$ is completely analogue.
  \end{proof}
\end{corollary}

\begin{corollary}
  \label{corr:Magnetisation}
  We have $M|\psi_{2n}\rangle = 0$ and $M|\psi_{2n-1}\rangle = (\ell/2)|\psi_{2n-1}\rangle$.
  \begin{proof}
The $E=0$ eigenspace of $H$ is one-dimensional. Furthermore the Hamiltonian conserves the magnetisation $[H,M]=0$. Hence, we must have $M|\psi_{L}\rangle = m_L|\psi_L\rangle$ for any $L\geqslant 1$. To find $m_L$, it is sufficient to project this equality onto simple basis vectors.
  For $L=2n$, we find
  \begin{equation}
    0 = \langle 0\ell \cdots 0\ell|M|\psi_{2n}\rangle = m_{2n} \langle 0\ell \cdots 0\ell|\psi_{2n}\rangle.
  \end{equation}
  According to \cref{corr:NonZeroComps}, $\langle 0\ell \cdots 0\ell |\psi_{2n}\rangle = (\psi_{2n})_{0\ell\cdots 0\ell}$ is non-zero and therefore $m_{2n}= 0$. 
  For $L=2n-1$ the proof is similar. 
  \end{proof}
\end{corollary}

For $L=2n$, the constraints \eqref{eqn:ConstraintsBasisVector} have many solutions. This allows us to relate certain components of the zero-energy state to the special component $(\psi_{2n})_{0\ell\cdots 0\ell}$.

\begin{proposition}
\label{prop:SpecialCompL2n}
For any weakly increasing sequence $p_1,\dots,p_n$ of integers, we have the component
\begin{equation}
 (\psi_{2n})_{p_1,\ell-p_1,\dots,p_n,\ell-p_n} =\frac{(\psi_{2n})_{0\ell\cdots 0\ell}}{\prod_{i=1}^n \{p_i+1\}}.
  \label{eqn:lambda2n}
\end{equation}
\begin{proof}
  It is trivial to check that \eqref{eqn:ConstraintsBasisVector} holds for $m_{2i-1}=p_i$ and $m_{2i}=\ell-p_i$ where $i=1,\dots,n$. Hence, using \eqref{eqn:SPPsi2n} and \eqref{eqn:DefChi} we find
  \begin{align}
     (\psi_{2n})_{p_1,\ell-p_1,\dots,p_n,\ell-p_n} &= (\psi_{2n})_{0\ell\cdots 0\ell}\langle p_1,\ell-p_1,\dots,p_n,\ell-p_n|\chit \cdots \chit\rangle
     \nonumber
     \\
     &= (\psi_{2n})_{0\ell\cdots 0\ell}\prod_{i=1}^n\chit_{p_i}=\frac{(\psi_{2n})_{0\ell\cdots 0\ell}}{\prod_{i=1}^n \{p_i+1\}},
  \end{align}
  which concludes the proof.
\end{proof}
\end{proposition}

The \cref{prop:CohomRepr,prop:SpecialCompL2n} imply that the zero-energy state is even under the action of the parity operator.

\begin{corollary}
  For any $L\geqslant 1$ we have $P|\psi_L\rangle = |\psi_L\rangle$.
  \begin{proof}
  The proof follows the lines of \cref{corr:Magnetisation}, using \cref{prop:SpecialCompL2n} in the case of even $L$.
  \end{proof}
\end{corollary}

\paragraph{Homology representation.} Up to now, we have focused on the cohomology of the supercharge $\Q$ and the resulting representations of the zero-energy states given in \cref{prop:CohomRepr}. The definition $H = \Q \Q^\dagger + \Q^\dagger \Q$ of a supersymmetric Hamiltonian suggests that we could as well have considered the adjoint supercharge $\Q^\dagger$. Its \textit{homology} is defined by
\begin{equation}
  \mathcal H_\bullet(\Q) = \bigoplus_{L=1}^\infty \mathcal H_L(\Q)
\end{equation}
where $\mathcal H_1(\Q^\dagger) = V/ \text{im} \{\Q^\dagger : V^{2} \to V\}$
and
\begin{equation}
    \mathcal H_L(\Q^\dagger) = \frac{\ker \{\Q^\dagger : V^L \to V^{L-1}\}}{\text{im} \{\Q^\dagger : V^{L+1} \to V^{L}\}}, \quad \text{for}\quad L\geqslant 2.
  \end{equation}
The existence of a non-degenerate inner product on $V^L$ implies that the cohomology and homology are isomorphic \cite{loday:92,masson:08}:
\begin{equation}
  \mathcal H_L(\Q^\dagger)\simeq \mathcal H^L(\Q)\quad \text{for each}\quad L\geqslant 1.
\end{equation}
In particular, we have
\begin{equation}
\dim \mathcal H_L(\Q^\dagger)=\dim \mathcal H^L(\Q)=1\quad  \text{for each}\quad L\geqslant 1.
\end{equation}
We now determine an alternative representation of the zero-energy states using this property. 

\begin{proposition}
  \label{prop:HomRepr}
  For each $n\geqslant 1$, we have
\begin{equation}
  \label{eqn:AdjointRepresentation}
  |\psi_{2n}\rangle = \mu_{2n} |0\ell\cdots 0\ell\rangle + \Q^\dagger |\tilde \phi_{2n}\rangle, \quad \mu_{2n} = \langle \chit\cdots \chit|\psi_{2n}\rangle,
  \end{equation}
  with $ |\tilde \phi_{2n}\rangle\in V^{2n+1}$ and
\begin{equation}
|\psi_{2n-1}\rangle = \mu_{2n-1} |0\ell\cdots 0\ell 0\rangle + \Q^\dagger |\tilde \phi_{2n-1}\rangle, \quad \mu_{2n-1} = \left(\langle 0|\otimes \langle \chit \cdots \chit|\right)|\psi_{2n-1}\rangle,
\end{equation}
with $|\tilde \phi_{2n-1}\rangle\in V^{2n}$.
\begin{proof}
We focus on the case $L=2n$. Since $\mathcal H_{2n}(\Q^\dagger)$ is one-dimensional, each of its non-zero elements can be represented by a non-zero multiple of a fixed vector in $V^{2n}$. This vector is in the kernel (but not in the image) of $\Q^\dagger$. We claim that such a vector is given by
    \begin{equation}
    |\omega\rangle = |0\ell\cdots 0\ell\rangle.
  \end{equation}
 Using \eqref{eqn:ActionAdjointTrigSC}, one readily checks that $|\omega\rangle$ is annihilated by $\Q^\dagger$. Furthermore, it cannot be in the image of $\Q^\dagger$. Indeed, otherwise if $|\omega\rangle = \Q^\dagger |\tilde \omega\rangle$ for some $|\tilde\omega\rangle \in V^{2n+1}$ then it follows that $(\psi_{2n})_{0\ell\cdots 0\ell}= \langle \omega|\psi_{2n}\rangle= \langle \tilde \omega|\Q|\psi_{2n}\rangle=0$. This contradicts \cref{corr:NonZeroComps}.
    
  It follows that the zero-energy states for $L=2n$ have the representation
  \begin{equation}
    |\psi_{2n}\rangle = \mu_{2n} |0\ell\cdots 0\ell\rangle + \Q^\dagger |\tilde \phi_{2n}\rangle
  \end{equation}
  for some non-zero scalar $\mu_{2n}$ and $|\tilde \phi_{2n}\rangle\in V^{2n+1}$. The vector $|\tilde \phi_{2n}\rangle$ cannot be determined from homological arguments. The factor $\mu_{2n}$, however, can be found by computing the scalar product of the zero-energy state with suitable states
   $|\gamma\rangle$ that are annihilated by the supercharge $\Q$ and have a non-zero scalar product with the representative. In the present case, we choose $|\gamma\rangle = |\chit\cdots \chit\rangle$. A short calculation, similar to \eqref{eqn:SPPsi2n}, leads to
    \begin{equation}
     \mu_{2n} = \langle \chit \cdots \chit|\psi_{2n}\rangle.
  \end{equation}
  
  The argument for $L=2n-1$ is similar, with the choice $\ket{\gamma}=\ket{0}\otimes\ket{\chit \dots \chit}$. 
\end{proof}
\end{proposition}

\subsection{Square norm and scalar products}

\label{sec:ScalarProducts}

In this subsection, we use the \cref{prop:CohomRepr,prop:HomRepr} in order to derive sum rules for the square norm of the zero-energy states and certain scalar products.

An immediate consequence of the (co)homology representations of the zero-energy states given in \cref{prop:CohomRepr,prop:HomRepr} is:
\begin{proposition}
For each $n\geqslant 1$, we have
\begin{subequations}
\label{eqn:SquareNorm}
\begin{equation}
    ||\psi_{2n}||^2 = (\psi_{2n})_{0\ell\cdots 0 \ell}\langle \psi_{2n}|\chit\cdots \chit\rangle
    \end{equation}
    and
\begin{equation}
    ||\psi_{2n-1}||^2 = (\psi_{2n-1})_{0\ell\cdots 0 \ell 0}\langle \psi_{2n-1}|\Bigl( |0\rangle\otimes |\chit\cdots \chit\rangle\Bigr).
\end{equation}
\end{subequations}
\begin{proof}
  For $L=2n$, we find
  \begin{align}
    ||\psi_{2n}||^2 &= (\mu_{2n}^\ast \langle 0\ell\cdots 0\ell| + \langle \tilde{\phi}_{2n}|\Q)|\psi_{2n}\rangle = \mu_{2n}^\ast (\psi_{2n})_{0\ell\cdots 0 \ell}\\
    &=(\psi_{2n})_{0\ell\cdots 0 \ell}\langle \psi_{2n}|\chit\cdots \chit\rangle \nonumber.
  \end{align}
  For $L=2n-1$ the proof is similar.
\end{proof}
\end{proposition}

The expressions for the square norm \eqref{eqn:SquareNorm} motivate the analysis of a family of scalar products involving the zero-energy states of arbitrary length. To be specific, let us consider $L=2n$. We find that
\begin{equation}
  \frac{\langle \psi_{2n}|\chit\cdots \chit\rangle}{||\psi_{2n}||} = \frac{1}{(\psi_{2n})_{0\ell\cdots 0 \ell}/||\psi_{2n}||}.
  \label{eqn:SimpleExample}
\end{equation}
On the left-hand side of this equality, we have the projection of the normalised zero-energy state onto an $n$-fold tensor product of zero-energy states of the two-site chain. On the right-hand site, we find the inverse of a special component of the normalised zero-energy state on $2n$ sites.

It is natural to investigate if the result remains equally simple when we replace the $|\chit\rangle$'s by the zero-energy energy states of spins chains of generic lengths. To this end, we cut an open chain of length $L$ into $m$ subchains of lengths $L_1,\dots,L_m>0$ as is illustrated in \cref{fig:Subdivision}.
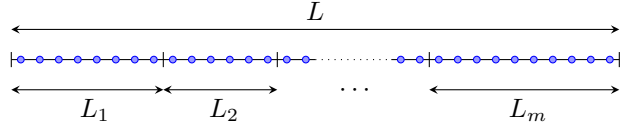
\begin{figure}[h]
  \centering
  \begin{tikzpicture}[>=stealth]
    \draw (0,0) -- (4,0);
    \draw[ dotted] (4,0) -- (5,0);
    \draw (5,0) -- (8,0);
    \foreach \x in {0,2,3.5,5.5,8}
    \draw[xshift = \x cm] (0,-.1) -- (0,.1);
    \draw[<->] (0,-.4) -- (2,-.4);
    \draw[<->] (2,-.4) -- (3.5,-.4);
    \draw[<->] (5.5,-.4) -- (8,-.4);
    \draw[<->] (0,.4) -- (8,.4);
    
    \foreach \x in {0,0.25,...,3.75}
      \draw[draw=blue,fill=blue!40] (\x+.125,0) circle (1.25pt);
    
    \foreach \x in {5,5.25,...,7.75}
      \draw[draw=blue,fill=blue!40] (\x+.125,0) circle (1.25pt);
    
    \draw (4.55,-0.4) node {$\cdots$};
    \draw (4,.4) node[above] {$L$};
    \draw (1.1,-.4) node[below] {$L_1$};
    \draw (2.8,-.4) node[below] {$L_2$};
    \draw (6.8,-.4) node[below] {$L_m$};

  \end{tikzpicture}
  \caption{Subdivision of an open chain of length $L$ into subchains.}
  \label{fig:Subdivision}
\end{figure}

Our aim is to compute projection of the zero-energy state of the complete chain onto the tensor product of the zero-energy states of these $m$ subchains. We consider thus the scalar product

\begin{equation}
   Z(L_1,\dots,L_m) = \frac{\langle \psi_L|\left(\bigotimes_{j=1}^m|\psi_{L_j}\rangle\right)}{||\psi_L||\prod_{j=1}^m ||\psi_{L_j}||}, \quad\text{for} \quad L= L_1+\dots+L_m .
   \label{eqn:MultipleOverlap}
\end{equation}
The division by the norms of the states makes this quantity normalisation-independent. For $L=2n$ and $L_j=2$ for $j=1,\dots,n$, we recover (up to factor) the left-hand side of \eqref{eqn:SimpleExample}.

\begin{theorem}
\label{thm:ResultForZ}
If $L_1,\dots,L_m$ are even, then
\begin{equation}
  Z(L_1,\dots,L_m) = \frac{\prod_{j=1}^m (\psi_{L_j})_{0\ell\cdots  0\ell }/ ||\psi_{L_j}||}{ (\psi_L)_{0\ell \cdots 0\ell }/ ||\psi_L||}.
  \label{eqn:Zeven}
\end{equation}
If $L_k$ is odd (for some $1\leqslant k\leqslant m$) and $L_1,\dots,L_{k-1},L_{k+1},\dots,L_m$ are even, then 
\begin{equation}
  Z(L_1,\dots,L_m) = \frac{(\psi_{L_k})_{0\ell 0\cdots \ell 0}/ ||\psi_{L_k}||\prod_{j=1,j\neq k}^m (\psi_{L_j})_{0\ell \cdots 0\ell }/ ||\psi_{L_j}||}{ (\psi_L)_{0\ell \cdots 0\ell 0}/ ||\psi_L||}.
  \label{eqn:Zoddeven}
\end{equation}
In all other cases, the scalar product vanishes.
\end{theorem}
\noindent
We note that the results \eqref{eqn:Zeven} and \eqref{eqn:Zoddeven} are remarkably simple. The zero-energy states typically are very complicated states with many non-zero components. Nonetheless, the scalar product can be inferred from the sole knowledge of a single component of each involved (normalised) state. Furthermore, the exchange of $L_i$ and $L_j$ for $i\neq j$ leaves the result invariant even though it can completely change the subdivision as shown in \cref{fig:Subdivision}.

\begin{proof}[Proof of \cref{thm:ResultForZ}]
 For simplicity, we focus on the proof of \eqref{eqn:Zeven}. Using \eqref{eqn:AdjointRepresentation}, we find
 \begin{equation}
   \langle \psi_L|\Biggl(\bigotimes_{j=1}^m|\psi_{L_j}\rangle\Biggr) = \mu_L^\ast \prod_{j=1}^m (\psi_{L_j})_{0\ell \cdots  0\ell}+ \langle \tilde{\phi}_L|\Q \Biggl(\bigotimes_{j=1}^m|\psi_{L_j}\rangle\Biggr).
 \end{equation}
The second term on the right-hand side of this equation vanishes. Indeed, we obtain
 \begin{equation} 
   \Q \Biggl(\bigotimes_{j=1}^m|\psi_{L_j}\rangle\Biggr) = \sum_{i=1}^m (-1)^{L_1+\dots+L_{i-1}}\bigotimes_{j=1}^{i-1}|\psi_{L_j}\rangle\otimes \Q|\psi_{L_i}\rangle\otimes \bigotimes_{j=i+1}^{m}|\psi_{L_j}\rangle 
 \end{equation}
 which is zero because $\Q|\psi_{L_i}\rangle=0$ for $i=1,\dots,m$. Hence, we find
 \begin{equation}
\langle \psi_L|\Biggl(\bigotimes_{j=1}^m|\psi_{L_j}\rangle\Biggr) = \mu_L^\ast \prod_{j=1}^m (\psi_{L_j})_{0\ell \cdots  0\ell}. 
 \end{equation}
Using \eqref{eqn:SquareNorm}, we write $\mu_L^\ast$ in terms of the square norm and the alternating component. Upon division by the norms of the zero-energy states, we obtain \eqref{eqn:Zeven}.
 
The proof of \eqref{eqn:Zoddeven} is similar. Eventually, the scalar product vanishes in all other cases because of the definite magnetisation of the zero-energy states.
 \end{proof}

There are many generalisations of \eqref{eqn:MultipleOverlap} that have equally simple expressions in terms of special components. As an example, we consider the product
\begin{subequations}
\begin{equation}
  \tilde Z(L_1,\dots, L_m) = \frac{\langle \psi_L|\left(|\psi_{L_1}\rangle\otimes  \bigotimes_{j=2}^{m}(|\ell\rangle \otimes |\psi_{L_j}\rangle)\right)}{||\psi_L||\prod_{j=1}^m ||\psi_{L_j}||}
\end{equation}
for
\begin{equation}
  L=L_1+\dots+L_m +m-1.
\end{equation}
\end{subequations}

\begin{theorem}
  \label{thm:ResultForZTilde}
  If $L_1,\dots,L_m$ are odd then
  \begin{equation}
  \tilde Z(L_1,\dots,L_m) = \frac{ (\psi_L)_{0\ell \cdots 0\ell 0}/ ||\psi_L||}{\prod_{j=1}^m (\psi_{L_j})_{0\ell\cdots  0\ell 0}/ ||\psi_{L_j}||}.
  \label{eqn:ZTildeodd}
  \end{equation}
  If $L_k$ is even (for some $1\leqslant k\leqslant m$) and $L_1,\dots,L_{k-1},L_{k+1},\dots,L_m$ are odd, then
\begin{equation}
 \tilde Z(L_1,\dots,L_m) =\frac{ (\psi_L)_{0\ell \cdots 0\ell}/ ||\psi_L||}{(\psi_{L_k})_{0\ell \cdots 0\ell}/ ||\psi_{L_k}||\prod_{j=1,j\neq k}^m (\psi_{L_j})_{0\ell \cdots 0\ell 0 }/ ||\psi_{L_j}||}.
  \label{eqn:ZTildeoddeven}
\end{equation}
  In all other cases, the scalar product vanishes.
\begin{proof}
Let us prove
\eqref{eqn:ZTildeodd}. The reasoning is similar to the proof of \cref{thm:ResultForZ}. We use \eqref{eqn:psi2n1} in order to write
\begin{equation}
  \label{eqn:ZTildeIntermediate}
\langle \psi_L|\Biggl(|\psi_{L_1}\rangle\otimes  \bigotimes_{j=2}^{m}(|\ell\rangle \otimes |\psi_{L_j}\rangle)\Biggr) = \left(\lambda_L^\ast \bra{0}\otimes
    \bra{\chit\cdots\chit} + \langle \phi_L|\Q^\dagger \right)\Biggl(|\psi_{L_1}\rangle\otimes  \bigotimes_{j=2}^{m}(|\ell\rangle \otimes |\psi_{L_j}\rangle)\Biggr).
\end{equation}
The term involving $\Q^\dagger$ vanishes. Indeed, one checks that if both $|\psi\rangle$ and $|\psi'\rangle$ are in the kernel of $\Q^\dagger$ then $\Q^\dagger(|\psi\rangle\otimes |\ell\rangle \otimes |\psi'\rangle)=0$ as a consequence of \eqref{eqn:ActionAdjointTrigSC}. Applying this property repeatedly, we reduce \eqref{eqn:ZTildeIntermediate} to 
\begin{equation}
\langle \psi_L|\Biggl(|\psi_{L_1}\rangle\otimes  \bigotimes_{j=2}^{m}(|\ell\rangle \otimes |\psi_{L_j}\rangle)\Biggr) = \lambda_L^\ast \prod_{j=1}^m \mu_{L_j}. 
\end{equation}
Using \eqref{eqn:SquareNorm}, we find \eqref{eqn:ZTildeodd}.

The proof of \eqref{eqn:ZTildeoddeven} is similar, using in addition the 
parity invariance of the zero-energy state. In all other cases, the scalar products vanish because of the definite magnetisation of the zero-energy states.
\end{proof}
\end{theorem}

It is possible to extend the definition of $\tilde Z$ in order to account for the case where $L_k=0$ for some $1\leqslant k\leqslant m$ and the $L_1,\dots,L_{k-1},L_{k+1},\dots,L_m$ are odd. In this case, we have
\begin{align}
  \tilde Z(L_1,\dots,L_k=0,\dots, L_m) &= \frac{\langle \psi_L|\left(  \bigotimes_{j=1}^{k-1}( |\psi_{L_j}\rangle\otimes |\ell\rangle)\otimes\bigotimes_{j=k+1}^{m}( |\ell\rangle\otimes |\psi_{L_j}\rangle )\right)}{||\psi_L||\prod_{j=1,j\neq k}^m ||\psi_{L_j}||}
  \nonumber
  \\
  &=\frac{ (\psi_L)_{0\ell \cdots 0\ell}/ ||\psi_L||}{\prod_{j=1,j\neq k}^m (\psi_{L_j})_{0\ell \cdots 0\ell 0 }/ ||\psi_{L_j}||}.
\end{align}

\section{The bipartite fidelity of the supersymmetric open XXZ chain}
\label{sec:SPXXZ}

In this section, we discuss the scalar products found in \cref{thm:ResultForZ,thm:ResultForZTilde} for $\ell=1$. In this case, the Hamiltonian is given by \eqref{eqn:XXZOpenChainND} and describes the open XXZ spin chain with anisotropy parameter $\Delta=-1/2$ and special diagonal boundary terms. In \cref{sec:FiniteSize}, we provide exact finite-size expressions for certain special components of the normalised zero-energy state. We use them to exactly evaluate the scalar products $Z(L_1,\dots,L_m)$ and $\tilde Z(L_1,\dots,L_m)$. Furthermore, we evaluate the scaling limits of the scalar products. In \cref{sec:LBF}, we compare these scaling limits for $m=2$ to the predictions of conformal field theory.

In contrast to the previous sections, the results presented here below are non-rigorous. Many statements are based on conjectures that we inferred from  the exact diagonalisation of the spin-chain Hamiltonian.

\subsection{Special components and scaling behaviour}

\label{sec:FiniteSize}
\paragraph{A finite-size conjecture.} We computed the zero-energy state $|\psi_L\rangle$ of the XXZ Hamiltonian \eqref{eqn:XXZOpenChainND} up to $L=16$ sites using \textsc{Mathematica}. We find that the ratio of the components $(\psi_{2n-1})_{01\cdots 010}$, $(\psi_{2n})_{01\cdots 01}$ and the norm of the corresponding states can be expressed in terms of two integer sequences $A_{\text{\rm \tiny V}}(2n+1)$ and $N_8(2n)$.
 These two sequences enumerate $(2n+1)\times (2n+1)$ vertically-symmetric alternating sign matrices and cyclically-symmetric self-complementary plane partitions in a $2n\times 2n \times 2n$ cube, respectively \cite{bressoudbook}. Explicitly, they are given by
\begin{equation}
  \label{eqn:DefAVN8}
  A_{\text{\rm \tiny V}}(2n+1)= \frac{1}{2^n}\prod_{k=1}^n \frac{(6k-2)!(2k-1)!}{(4k-1)!(4k-2)!}, \quad N_{8}(2n)=\prod_{k=0}^{n-1} \frac{(3k+1)(6k)!(2k)!}{(4k)!(4k+1)!}.
\end{equation}
Our numerical investigation is consistent with the following conjecture: 
\begin{conjecture}
  \label{conj:NormalisedComponents}
For each $n\geqslant 1$ we have 
\begin{equation}
  \label{eqn:NormalisedComponents}
   \frac{(\psi_{2n-1})_{01\cdots 010}}{||\psi_{2n-1}||} = \sqrt{\frac{N_8(2n)}{A_{\text{\rm \tiny V}}(2n+1)}}
, \quad  \frac{(\psi_{2n})_{01\cdots 01}}{||\psi_{2n}||} = \sqrt{\frac{A_{\text{\rm \tiny V}}(2n+1)}{N_8(2n+2)}}.
\end{equation}
\end{conjecture}

The occurrence of sequences that enumerate alternating sign matrices and plane partitions in the ground state of an XXZ spin chain at $\Delta=-1/2$ comes by no means as a surprise (see for example \cite{razumov:00,razumov:01,batchelor:01,degier:02,difrancesco:06,razumov:07}). We provide a proof of \cref{conj:NormalisedComponents} along with many other combinatorial properties of the ground state $|\psi_L\rangle$ in a separate publication \cite{hagendorf:tbp}.

\paragraph{Scaling behaviour.} The sequences $A_{\textrm{\tiny V}}(2n+1)$ and $N_{8}(2n)$ are given by ratios of products of factorials. Hence, they can be expressed in terms of Barnes' $G$-function \cite{whittaker:27}. We use the well-known asymptotic expansion of this function to evaluate the components \eqref{eqn:NormalisedComponents} for large system sizes:
\begin{subequations}
\begin{align}
  \frac{(\psi_{2n-1})_{01\cdots 010}}{||\psi_{2n-1}||} & = C_1(2n)^{1/12}\left(\frac{3^{3/4}}{2}\right)^{-2n}\left(1+O(n^{-1})\right),\\\frac{(\psi_{2n})_{01\cdots 01}}{||\psi_{2n}||} &= C_2(2n)^{-1/12}\left(\frac{3^{3/4}}{2}\right)^{-2n}\left(1+O(n^{-1})\right).
\end{align}
Here $C_1$ and $C_2$ are the constants 
\begin{equation}
   C_1 = \frac{\sqrt{\Gamma(1/3)}}{\pi^{1/4}}, \quad C_2=\left(\frac{2}{\sqrt{3}}\right)^{3/2} \frac{\pi^{1/4}}{\sqrt{\Gamma(1/3)}}.
\end{equation}
\label{eqn:asymptotics}%
\end{subequations}%
We use these expressions to extract the scaling behaviour of $Z(L_1,\dots,L_m)$, for $\ell=1$. It is obtained when the lengths of the subintervals $L_1,\dots,L_m$ become large in such a way that the ratios $L_i/L$ approach certain scaling variables $0<x_i<1$ for each $i=1,\dots,m$. $Z(L_1,\dots,L_m)$ is then given by an asymptotic series with respect to the system size $L$. The series coefficients are functions of $x_1,\dots,x_m$. Notice that together, \cref{thm:ResultForZ} and \eqref{eqn:asymptotics} imply that this asymptotic series is only well defined if the parity of the integers $L_1,\dots,L_m$ is fixed. 
There are two interesting cases, corresponding to \eqref{eqn:Zeven} and \eqref{eqn:Zoddeven}:
\begin{subequations}
\begin{enumerate}[label=\textit{(\roman*)}]
\item $L_i$ is even for each $i=1,\dots,m$. In this case the length of the chain $L$ is even. Using \eqref{eqn:asymptotics}, we obtain
\begin{equation}
  \label{eqn:ZScalingEven}
  Z(L_1,\dots,L_m) = L^{-(m-1)/12}C_2^{m-1}\prod_{i=1}^{m}x_i^{-1/12}\left(1+O(L^{-1})\right).
\end{equation}
\item $L_j$ is odd for a certain $j$ and $L_i$ is even for each $i=1,\dots,j-1,j+1,\dots,m$. In this case, $L$ is odd. From \eqref{eqn:asymptotics}, we obtain 
\begin{equation}
 \label{eqn:ZScalingEvenOdd}
 Z(L_1,\dots,L_m)=L^{-(m-1)/12}C_2^{m-1}x_j^{1/12}\prod_{i=1, i\neq j}^{m}x_i^{-1/12}\left(1+O(L^{-1})\right).
\end{equation}
\end{enumerate}

We obtain the scaling behaviour of the scalar product $\tilde Z(L_1,\dots,L_m)$ from \cref{thm:ResultForZTilde} and \eqref{eqn:asymptotics}. The resulting leading-order terms of \eqref{eqn:ZTildeodd} and \eqref{eqn:ZTildeoddeven} are up to a factor equal to \eqref{eqn:ZScalingEven} and \eqref{eqn:ZScalingEvenOdd}, respectively. Hence, without loss of generality we focus in the following on $Z(L_1,\dots,L_m)$.

\label{eqn:ScalingLimitsZ}
\end{subequations}

\subsection{Scaling behaviour and conformal field theory}

\label{sec:LBF}

The power-law decay of the scalar products in \eqref{eqn:ScalingLimitsZ} as well as their algebraic dependence on the scaling coordinates $x_1,\dots,x_{m}$ suggest that they could be related to correlation functions of conformal field theory (CFT). In this section, we discuss this relation for the special case of $m=2$. Specifically, we consider the so-called \textit{logarithmic bipartite fidelity} (LBF):
\begin{equation}
  \mathcal F(L_1,L_2) = - \ln |Z(L_1,L_2)|^2.
\end{equation}
It was introduced by Dubail and St\'ephan as an entanglement measure for the ground state of interacting quantum many-body systems in one dimension\cite{dubail:11,dubail:13}. In particular, they predicted the leading-order terms of the asymptotic expansion of the LBF with respect to the system size $L$ for one-dimensional quantum critical systems from CFT arguments.

From \eqref{eqn:ZScalingEven} and \eqref{eqn:ZScalingEvenOdd} we exactly compute these leading-order terms. We express them in terms of $x=x_1$. In case \textit{(i)}, we obtain
\begin{subequations}
\begin{equation}
  \mathcal F = \frac{1}{6}\ln L + \frac{1}{6}\ln x(1-x)-2\ln C_2+O(L^{-1}).
  \label{eqn:LBFScalingEven}
\end{equation}
In the case \textit{(ii)}, we take the scaling limit with $L_1$ even and $L_2$ odd. This leads to
 \begin{equation}
  \mathcal F = \frac{1}{6}\ln L + \frac{1}{6}\ln \left(\frac{x}{1-x}\right)-2\ln C_2+O(L^{-1}).
  \label{eqn:LBFScalingOdd}
\end{equation}
\label{eqn:LBFScaling}%
\end{subequations}%
The aim of this section is to show that \eqref{eqn:LBFScaling} perfectly matches the CFT predictions. To this end, we briefly discuss the relation between the scaling limit of the open XXZ chain \eqref{eqn:XXZOpenChainND} at $\Delta=-1/2$ and superconformal CFT in \cref{sec:cft}. In \cref{sec:lbf} we compare our findings for the scaling limit of the spin chain's LBF to the CFT results.

\subsubsection{CFT connection}

\label{sec:cft}
It is well known that in suitable scaling limits, many properties of the XXZ spin with anisotropy parameter $-1\leqslant \Delta \leqslant 1$ are accurately described by CFT. For open chains and real diagonal boundary terms, the field theory is expected to be given by a free boson theory with central charge $c=1$ and a compactification radius that depends on the value of $\Delta$ \cite{affleck:90}. If $\Delta=-1/2$ then this compactification radius is fixed to a value where the field theory coincides with the first model of the so-called $\mathcal N=2$ superconformal minimal series \cite{waterson:86,friedan:88,lerche:89} (see also \cite{huijse:10} for a compact introduction accessible to non-experts). This implies in particular that the CFT space of states divides into the so-called Ramond and Neveu-Schwarz sectors. Each sector organises in a finite number of irreducible highest-weight representations of the so-called $\mathcal N=2$ superconformal algebra. From the representation theory of this algebra, we know that the corresponding highest-weight states are labeled by a pair $(h,\alpha)$. Here $h$ denotes the conformal weight and $\alpha$ the so-called $U(1)$ charge of the state.\footnote{We use the common notation $h$ for conformal weights. This is not to be confused with the Hamiltonian density discussed in previous sections. Furthermore, for the $U(1)$-charges $\alpha$, we follow the conventions of \cite{dubail:13} that differ from the standard choice in the literature on superconformal field theory. For the first $\mathcal N=2$ superconformal minimal model that choice is $q_{\text{\tiny SCFT}}=\alpha/\sqrt{3}$.}

In the following, we focus on the Ramond sector. For the first $\mathcal N=2$ superconformal minimal model, it contains three highest-weight representations. Two of them are singlet representations with the corresponding highest-weight states labeled by
\begin{subequations}
\begin{equation}
  \left(\frac{1}{24},\frac{1}{2\sqrt{3}}\right),\,\left(\frac{1}{24},-\frac{1}{2\sqrt{3}}\right).
\end{equation}
Furthermore, there is a doublet representation with a highest-weight state labeled by
\begin{equation}
  \left(\frac{3}{8},\frac{\sqrt{3}}{2}\right)
\end{equation}
\label{eqn:CFTData}%
\end{subequations}%
It has a superpartner with the same conformal weight but the opposite $U(1)$-charge: $(3/8,-\sqrt{3}/2)$. Notice that since we deal with a free boson theory we have $h=\alpha^2/2$ for all these states \cite{difrancesco:97}.

\paragraph{Lattice and CFT quantities.}
The XXZ Hamiltonian \eqref{eqn:XXZOpenChainND} possesses an explicit lattice realisation of the supersymmetry. This suggests that for finite $L$ its low-energy eigenstates constitute an approximation to the Ramond sector of the CFT Hilbert space. In order to substantiate this claim, we identify the representation data \eqref{eqn:CFTData} with certain properties of the ground state and first excited state of the lattice Hamiltonian (see also \cite{degier:05} and furthermore \cite{huijse:11_2} for a related analysis of the $M_1$ model of supersymmetric fermions).

The conformal weight $h$ of a low-energy eigenstate of the Hamiltonian is related to the finite-size scaling of its eigenvalue. Let us denote by $E_0(L) \leqslant E_1(L) \leqslant E_2(L) \leqslant \cdots$ the eigenvalues of the Hamiltonian at size $L$. Then we have for large $L$ and small $i$ the following expansion \cite{cardy:86_3,affleck:86}:
\begin{equation}
  E_i(L) = L E_{\text{bulk}}+E_{\text{bdr}}+\frac{\pi v_F}{L}\left(h_i-\frac{c}{24}\right)+O(L^{-2}),
  \label{eqn:FiniteSizeScaling}
\end{equation}
where the central charge takes the value $c=1$. The factor $v_F$ is the so-called Fermi velocity that can be computed by Bethe-ansatz techniques \cite{korepin:93}. For $\Delta=-1/2$ we have $v_F = 3\sqrt{3}/2$. Furthermore, $E_{\text{bulk}}$ and $E_{\text{bdr}}$ are non-universal constants that depend on the definition of the Hamiltonian. In our case, these constants have to be zero, $E_{\text{bulk}}=E_{\text{bdr}}=0$, because the spectrum contains the ground-state eigenvalue $E_0(L)=0$ for any $L\geqslant 2$. Furthermore, \eqref{eqn:FiniteSizeScaling} fixes the conformal weight of the zero-energy states to $h_0=1/24$. To probe the consistency of these assignments, we have numerically examined the dependence of the eigenvalue $E_1(L)$ on $L$. Its scaling behaviour matches well \eqref{eqn:FiniteSizeScaling} with $h_1=3/8$.

Furthermore, it is natural to assume that the $U(1)$-charge of an eigenstate of the Hamiltonian is related to a linear function of its magnetisation. We fix this function by comparing the $U(1)$-charges of the two Ramond ground states in \eqref{eqn:CFTData} and the magnetisation of the lattice ground states that we obtained in \cref{corr:Magnetisation}. This suggests that the operator
\begin{equation}
  J_0 = \frac{(1-4M)}{2\sqrt{3}}.
\end{equation}
measures the CFT $U(1)$-charge: we have
\begin{equation}
   J_0|\psi_{2n-1}\rangle = -\frac{1}{2\sqrt{3}}|\psi_{2n-1}\rangle \quad \text{and} \quad  J_0|\psi_{2n}\rangle = \frac{1}{2\sqrt{3}}|\psi_{2n}\rangle
   \label{eqn:U1ZeroEnergyStates}
\end{equation}
for each $n\geqslant 1$. To see that this choice is consistent, we inspect the first excited state of the Hamiltonian. We denote this state by $|\phi_{L}\rangle$ for a system of length $L$. Our exact diagonalisation data supports the conjecture that for each $n\geqslant 1$, we have $\Q^\dagger |\phi_{2n-1}\rangle =0$ and $\Q|\phi_{2n-1}\rangle = |\mathcal \phi_{2n}\rangle$. Furthermore, the magnetisation of the states that we computed for small systems is compatible with the conjecture
\begin{equation}
   J_0|\phi_{2n-1}\rangle = \frac{\sqrt{3}}{2}|\phi_{2n-1}\rangle \quad \text{and} \quad  J_0|\phi_{2n}\rangle = - \frac{\sqrt{3}}{2}|\phi_{2n}\rangle.
\end{equation}
for each $n\geqslant 1$. This suggests that the first excited state for spin chains of odd length is a lattice approximation to the CFT highest-weight state $(3/8,\sqrt{3}/2)$.
 
 \paragraph{Conformal weight of the cut.} In order to compare our results \eqref{eqn:LBFScalingEven} and \eqref{eqn:LBFScalingOdd} to the predictions from conformal field theory, we need a last ingredient: the $U(1)$-charge of the cut. To motivate the introduction of this quantity, recall that the scalar product $Z(L_1,L_2)$ is non-zero only if the product-state $|\psi_{L_1}\rangle\otimes |\psi_{L_2}\rangle$ and $|\psi_L\rangle$ have the same magnetisation. Let us now rephrase this statement in terms of the $U(1)$-charges of the three zero-energy states $|\psi_{L_1}\rangle$, $|\psi_{L_2}\rangle$ and $|\psi_L\rangle$, which we denote by $\alpha_1,\alpha_2$ and $\alpha_3$, respectively. We find that $Z(L_1,L_2)$ is non-vanishing only if
\begin{equation}
  \alpha_\text{c} \equiv \alpha_3-\alpha_1-\alpha_2 =-\frac{1}{2\sqrt{3}}.
\end{equation}
The quantity $\alpha_\text{c}$ is the $U(1)$-charge of the cut. It allows us to formally define the conformal weight of the cut $h_\text{c}=\alpha_\text{c}^2/2=1/24$. Within CFT, it is possible to identify $h_\text{c}$ with the conformal weight of a so-called boundary condition changing operator. We refer to \cite{dubail:13} for more details.
 
\subsubsection{Logarithmic bipartite fidelity}

\label{sec:lbf}

\paragraph{Finite-size scaling of the LBF.} For quantum critical systems, the leading-order terms of the asymptotic expansion of the LBF with respect to the system size $L$ is given by \cite{dubail:13}:
\begin{subequations}
\begin{equation}
  \mathcal F = \left(\frac{c}{8}+h_{\text{c}}\right)\ln L + f(x) + g(x) L^{-1}\ln L+ O(L^{-1}).
 \label{eqn:LBFCFT}
\end{equation}
Here, $c$ is the central charge of the theory and $h_{\text{c}}$
 the conformal weight of the cut. Furthermore, $x=x_1$ denotes the scaling variable defined above. The functions $f(x)$ and $g(x)$ depend on the conformal weights and $U(1)$-charges associated to the involved states and the cut. Their explicit general form can be found in \cite{dubail:13}. Here, we only consider the case of a $c=1$ free field theory, where they can be written in terms of the $U(1)$ charges alone:
\begin{align}
f(x) =& \left(\frac{1}{24}\left(2x-1+\frac{2}{x}\right)+\left(1-\frac{1}{x}\right)\alpha_1^2 +(1-x)\alpha_3^2 -\frac{\alpha_\text{c}^2}{2}-\alpha_2^2-2\alpha_\text{c}\alpha_2\right)\ln (1-x)
\nonumber\\
&+ \left\{x\to 1-x; \,\alpha_1\leftrightarrow \alpha_2\right\} + C
\label{eqn:fFromCFT}
\end{align}
and
\begin{equation}
  g(x) = \xi\times \frac{1}{2}\left(\alpha_3^2 - \frac{1}{12} + \left(\frac{1}{12}-\alpha_1^2\right)\frac{1}{x}+\left(\frac{1}{12}-\alpha_2^2\right)\frac{1}{1-x}\right).
  \label{eqn:gFromCFT}%
\end{equation}%
\label{eqn:CFTPrediction}%
\end{subequations}%
Here $C$ and $\xi > 0$ are non-universal constants. In particular, $\xi$ is called the \textit{extrapolation length} \cite{diehl:81}.
Apart from these two constants, the expressions are universal.

\paragraph{Comparison to scaling limit.} The identification of the $U(1)$-charges of the zero-energy states and the cut allow us to finally compare our findings \eqref{eqn:LBFScaling} with the CFT prediction \eqref{eqn:CFTPrediction}. First, we notice that since $c=1$ and $h_\text{c}=1/24$, the leading order term in \eqref{eqn:LBFCFT} is $\frac16 \ln L$. This is consistent with our findings in \eqref{eqn:LBFScaling}. Second, for the subleading terms we find the following results. In case \textit{(i)} we have
\begin{equation}
  \alpha_1=\alpha_2=\alpha_3 = \frac{1}{2\sqrt{3}}, \, \alpha_\text{c} = -\frac{1}{2\sqrt{3}}, \quad f(x) =\frac{1}{6}\ln x(1-x)+C,\quad g(x) = 0.
\end{equation}
In case \textit{(ii)} we have
\begin{equation}
  \alpha_2=\alpha_3=\alpha_\text{c}=-\frac{1}{2\sqrt{3}}, \, \alpha_1=\frac{1}{2\sqrt{3}}, \quad f(x) =\frac{1}{6}\ln \left(\frac{x}{1-x}\right)+C,\quad g(x) = 0.
\end{equation}
In both cases, the expressions for $f(x)$ and $g(x)$ match perfectly the scaling limit of the lattice results, provided that we set $C=-2\ln C_2$. The absence of the $O(L^{-1}\ln L)$-term is worth mentioning. Similar instances of vanishing finite-size corrections for the quantum spin $1/2$ XXZ chain at $\Delta=-1/2$ were reported in the literature \cite{banchi:09}. In the present case, the vanishing is a direct consequence of the supersymmetry as it fixes the $U(1)$-charges of the zero-energy states to $\pm 1/(2\sqrt{3})$.

We end our analysis with a comment on the higher-spin cases $\ell>1$. For these cases we expect the $O(L^{-1}\ln L)$ term to vanish, too. Indeed, for arbitrary $\ell$ the scaling limit of the spin chains
  should be described by the Ramond sector of the $\ell$-th $\mathcal N=2$ superconformal minimal model. Indeed, this conjecture was made for the related $M_\ell$ models of strongly-interacting fermions with supersymmetry \cite{fendley:03}. It is well-known that the Ramond ground states have the conformal weight $h=c/24$ . Furthermore, for arbitrary $c$ the scaling function $g(x)$ is given by \cite{dubail:13}
\begin{equation}
  g(x) = \xi\times \left(h_3 - \frac{c}{24} + \left(\frac{c}{24}-h_1\right)\frac{1}{x}+\left(\frac{c}{24}-h_2\right)\frac{1}{1-x}\right), \quad \xi > 0.
\end{equation}%
If all of the involved states are Ramond ground states (or at least lattice approximations thereof) then this expression identically vanishes.

\section{Conclusion}

\label{sec:Conclusion}

In this article, we have studied the dynamic lattice supersymmetry for the open XXZ Heisenberg chains at their supersymmetric point. We have determined the family $\Q_{j,k}(y)$ of supercharges and identified a family of non-diagonal boundary interactions that are compatible with the supersymmetry. Furthermore, we have computed the cohomology of the supercharges and shown that it is non-trivial if and only if $y=0$. From this cohomology computation we have deduced the existence of a zero-energy state, unique up to normalisation, of the spin-chain Hamiltonians with $y=0$. A central result of this article is the sum rules presented in \cref{thm:ResultForZ,thm:ResultForZTilde}: They reveal that a large family of scalar products involving an arbitrary number of normalised zero-energy states can be computed in terms of certain distinguished components. We emphasise that these results are solely based on the supersymmetric structure and do not make use of any integrability techniques. Eventually, we have computed the scaling behaviour of logarithmic bipartite fidelity for the open spin $1/2$ quantum XXZ chain at $\Delta=-1/2$ with special boundary magnetic fields. Its scaling behaviour matches the predictions from conformal field theory at both leading and sub-leading orders.

The results of this work have a few interesting generalisations. First, it is natural to consider closed spin chains with periodic or twisted boundary conditions. In this case, the dynamic lattice supersymmetry only exists in certain (anti-)cyclic subspaces of the spin-chain Hilbert space \cite{hagendorf:13}. Therefore, the computation of the cohomology of the supercharges is more challenging. Nonetheless, we expect that there is a connection between the cohomology of supercharges of the open and the closed chains, based on known examples in the mathematical literature \cite{loday:92}. It should allow to establish some generalisations of \cref{thm:ResultForZ,thm:ResultForZTilde} to closed spin chains for the class of models discussed in this article.
 Second, the computation of similar scalar products in off-critical models with lattice supersymmetry, in particular their scaling limit in the vicinity of a critical point, is of potential interest. An example is the open staggered $M_1$ model \cite{beccaria:12}. We hope to address these problems in forthcoming publications.

\subsection*{Acknowledgements}

This work is supported by the Belgian Interuniversity Attraction Poles Program P7/18 through the network DYGEST (Dynamics, Geometry and Statistical Physics). We thank Alexi Morin-Duchesne for many interesting remarks and stimulating suggestions. Furthermore, we thank Luigi Cantini, J\'er\^ome Dubail, Axel Marcillaud de Goursac and Anita Ponsaing for discussions.


\begin{thebibliography}{10}

\bibitem{bethe:31}
H.~Bethe,
\newblock {\em {Zur Theorie der Metalle. I. Eigenwerte und Eigenfunktionen der
  linearen Atomkette}},
\newblock Z. Phys. {\textbf{71}} (1931)   205---226.

\bibitem{korepin:93}
V.E.~Korepin, N.M.~Bogoliubov  and A.G.~Izergin,
\newblock {\em {Quantum Inverse Scattering Method and Correlation Functions.}},
\newblock Cambridge University Press (1993).

\bibitem{maillet:00_2}
J.-M.~Maillet and V.~Terras,
\newblock {\em On the quantum inverse scattering problem},
\newblock Nuclear Physics B {\textbf{575}} {\textbf{3}} (2000)   627 -- 644.

\bibitem{jimbo:94}
M.~Jimbo and T.~Miwa,
\newblock {\em Algebraic analysis of solvable lattice models}, {\em {\em
  Volume}~85},
\newblock American Mathematical Soc. (1994).

\bibitem{maillet:07}
J.-M.~Maillet,
\newblock {\em Heisenberg spin chains: from quantum groups to neutron
  scattering experiments},
\newblock S\'eminaire Poincar\'e {\textbf{X}} {\textbf{139--177}} (2007).

\bibitem{schollwoeck:08}
U.~Schollw{\"o}ck, J.~Richter, D.~Farnell  and R.~F. Bishop,
\newblock {\em Quantum magnetism}, {\em {\em Volume} 645},
\newblock Springer (2008).

\bibitem{kirillov:87}
A.N.~{Kirillov} and N.Y.~{Reshetikhin},
\newblock {\em {Exact solution of the integrable XXZ Heisenberg model with
  arbitrary spin. I. The ground state and the excitation spectrum}},
\newblock J. Phys. A : Math. Gen. {\textbf{20}} (1987)   1565--1585.

\bibitem{witten:82}
E.~Witten,
\newblock {\em {Constraints on supersymmetry breaking}},
\newblock Nucl. Phys. B {\textbf{202}} (1982)   253 -- 316.

\bibitem{yang:04}
X.~{Yang} and P.~{Fendley},
\newblock {\em {Non-local spacetime supersymmetry on the lattice}},
\newblock J. Phys. A: Math. Gen. {\textbf{37}} (2004)   8937--8948.

\bibitem{beisert:04}
N.~Beisert,
\newblock {\em The $su(2|3)$ dynamic spin chain},
\newblock Nucl. Phys. B {\textbf{682}} (2004)   487 -- 520.

\bibitem{hagendorf:12}
C.~Hagendorf and P.~Fendley,
\newblock {\em {The eight-vertex model and lattice supersymmetry}},
\newblock J. Stat. Phys. {\textbf{146}} (2012)   1122--1155.

\bibitem{hagendorf:13}
C.~Hagendorf,
\newblock {\em {Spin chains with dynamical lattice supersymmetry}},
\newblock J. Stat. Phys. {\textbf{150}} (2013)   609--657.

\bibitem{meidinger:14}
D.~Meidinger and V.~Mitev,
\newblock {\em Dynamic lattice supersymmetry in $\mathfrak{gl}(m|n)$ spin
  chains},
\newblock J. Stat. Phys. {\textbf{156}} (2014)   1199--1220.

\bibitem{dubail:11}
J.~Dubail and J.-M.~St\'ephan,
\newblock {\em {Universal behavior of a bipartite fidelity at quantum
  criticality}},
\newblock J. Stat. Mech. (2011)   L03002.

\bibitem{dubail:13}
J.-M.~St\'ephan and J.~Dubail,
\newblock {\em Logarithmic corrections to the free energy from sharp corners
  with angle $2\pi$},
\newblock J. Stat. Mech.  {\textbf{09}} (2013)   P09002.

\bibitem{weston:11}
R.~Weston,
\newblock {\em {Correlation functions and the boundary qKZ equation in a
  fractured XXZ chain}},
\newblock J. Stat. Mech (2011)   P12002.

\bibitem{weston:12}
R.~Weston,
\newblock {\em {Exact and scaling form of the bipartite fidelity of the
  infinite XXZ chain}},
\newblock J. Stat. Mech.  {\textbf{04}} (2012)   L04001.

\bibitem{fendley:03}
P.~{Fendley}, B.~{Nienhuis}  and K.~{Schoutens},
\newblock {\em {Lattice fermion models with supersymmetry}},
\newblock J. Phys. A: Math. Gen. {\textbf{36}} (2003)   12399--12424.

\bibitem{beccaria:05}
M.~{Beccaria} and G.F.~{de Angelis},
\newblock {\em {Exact Ground State and Finite-Size Scaling in a Supersymmetric
  Lattice Model}},
\newblock Phys. Rev. Lett. {\textbf{94}} {\textbf{10}} (2005)   100401.

\bibitem{nichols:05}
A.~{Nichols}, V.~{Rittenberg}  and J.~{de Gier},
\newblock {\em {One-boundary Temperley Lieb algebras in the XXZ and loop
  models}},
\newblock J. Stat. Mech. (2005)   P03003.

\bibitem{degier:05}
J.~de~Gier, A.~Nichols, P.~Pyatov  and V.~Rittenberg,
\newblock {\em {Magic in the spectra of the XXZ quantum chain with boundaries
  at $\Delta=0$ and $\Delta=-1/2$}},
\newblock Nucl. Phys. B {\textbf{729}} (2005)   387--418.

\bibitem{pasquier:90}
V.~Pasquier and H.~Saleur,
\newblock {\em {Common structures between finite systems and conformal field
  theories through quantum groups}},
\newblock Nucl. Phys. B {\textbf{330}} (1990)   523 -- 556.

\bibitem{morin:16}
A.~{Morin-Duchesne}, J.~{Rasmussen}, P.~{Ruelle}  and Y.~{Saint-Aubin},
\newblock {\em {On the reality of spectra of U$_{q}$(sl$_{2}$)-invariant XXZ
  Hamiltonians}},
\newblock J. Stat. Mech. (2016)   P053105.

\bibitem{zamolodchikov:81}
A.B.~Zamolodchikov and V.A.~Fateev,
\newblock {\em {A model factorized $S$-matrix and an integrable spin-$1$
  Heisenberg chain}},
\newblock Sov. J. Nucl. Phys. {\textbf{32}} (1981)   298--303.

\bibitem{mangazeev:14}
V.V.~Mangazeev,
\newblock {\em {On the Yang-Baxter equation for the six-vertex model}},
\newblock Nucl. Phys. B {\textbf{882}} (2014)   70 -- 96.

\bibitem{loday:92}
J.-L.~Loday,
\newblock {\em Cyclic homology},
\newblock Springer-Verlag Berlin Heidelberg, 1992.

\bibitem{huijse:10_1}
L.~Huijse and K.~Schoutens,
\newblock {\em {Supersymmetry, lattice fermions, independence complexes and
  cohomology theory}},
\newblock Adv. Theor. Math. Phys. {\textbf{14}} (2010)   643--694.

\bibitem{huijse:15}
L.~{Huijse} and C.~{Hagendorf},
\newblock \textit{On the ground states of the $M_\ell$ models},
\newblock arXiv:1509.08879 (2015).

\bibitem{masson:08}
T.~Masson,
\newblock {\em {Introduction aux (Co)Homologies}},
\newblock \'Editions Hermann, Paris (2008).

\bibitem{bressoudbook}
D.~Bressoud,
\newblock {\em {Proofs and confirmations: the story of the alternating sign
  matrix conjecture}},
\newblock Cambridge University Press (1999).

\bibitem{razumov:00}
A.V.~{Razumov} and Y.G.~{Stroganov},
\newblock {\em {Spin chains and combinatorics}},
\newblock J. Phys. A : Math. Gen. {\textbf{34}} (2001)  3185--3190.

\bibitem{razumov:01}
A.V.~{Razumov} and Y.G.~{Stroganov},
\newblock {\em {Spin chains and combinatorics: twisted boundary conditions}},
\newblock J. Phys. A: Math. Gen. {\textbf{34}} (2001)   5335--5340.

\bibitem{batchelor:01}
M.T. {Batchelor}, J.~{de Gier}  and B.~{Nienhuis},
\newblock {\em {The quantum symmetric XXZ chain at {$\Delta = - 1/2$},
  alternating-sign matrices and plane partitions}},
\newblock J. Phys. A: Math. Gen. {\textbf{34}} (2001)   L265--L270.

\bibitem{degier:02}
J.~{de Gier}, M.T.~{Batchelor}, B.~{Nienhuis}  and S.~{Mitra},
\newblock {\em {The XXZ spin chain at {$\Delta=-1/2$}: Bethe roots, symmetric
  functions, and determinants}},
\newblock J. Math. Phys. {\textbf{43}} (2002)   4135--4146.

\bibitem{difrancesco:06}
P.~{Di Francesco}, P.~{Zinn-Justin}  and {J.-B.} {Zuber},
\newblock {\em {Sum rules for the ground states of the O(1) loop model on a
  cylinder and the XXZ spin chain}},
\newblock J. Stat. Mech. {\textbf{8}} (2006)  11.

\bibitem{razumov:07}
A.V.~Razumov, Y.G.~Stroganov  and P.~Zinn-Justin,
\newblock {\em {Polynomial solutions of qKZ equation and ground state of XXZ
  spin chain at $\Delta = -1/2$}},
\newblock J. Phys. A : Math. Gen. {\textbf{40}} (2007)   11827.

\bibitem{hagendorf:tbp}
C.~Hagendorf and J.~Li\'enardy.
\newblock To be published,
\newblock (2017).

\bibitem{whittaker:27}
E.T.~Whittaker and G.N.~Watson,
\newblock {\em A course of modern analysis},
\newblock Cambridge University Press (1927).

\bibitem{affleck:90}
I.~Affleck,
\newblock {\em Field theory methods and quantum critical phenomena},
\newblock in E.~Br\'ezin and J.~Zinn-Justin, editors, {\em Fields, Strings and
  Critical Phenomena}, North Holland, Amsterdam, (1990) 563--640.

\bibitem{waterson:86}
G.~Waterson,
\newblock {\em {Bosonic construction of an $N = 2$ extended superconformal
  theory in two dimensions}},
\newblock Phys. Lett. B {\textbf{171}} (1986)   77 -- 80.

\bibitem{friedan:88}
D.~{Friedan} and S.~{Shenker},
\newblock {\em Supersymmetric critical phenomena and the two-dimensional
  gaussian model},
\newblock in C.~Itzykson, H.~Saleur  and J.-B.~Zuber, eds., {\em Conformal
  invariance and applications to statistical mechanics}, World
  Scientific Publishing, (1988) 578--593.

\bibitem{lerche:89}
W.~Lerche, C.~Vafa  and N.P.~Warner,
\newblock {\em {Chiral rings in $N=2$ superconformal field theories}},
\newblock Nucl. Phys. B {\textbf{324}} (1989)   427--474.

\bibitem{huijse:10}
L.~Huijse,
\newblock {\em A supersymmetric model for lattice fermions},
\newblock PhD thesis, Universiteit van Amsterdam (2010).

\bibitem{difrancesco:97}
P.~Di Francesco, P.~Mathieu  and D.~S\'en\'echal,
\newblock {\em {Conformal Field Theory}},
\newblock Springer (1997).

\bibitem{huijse:11_2}
L.~{Huijse},
\newblock {\em {Detailed analysis of the continuum limit of a supersymmetric
  lattice model in 1D}},
\newblock J. Stat. Mech. (2011)   P04004.

\bibitem{cardy:86_3}
H.W.J.~Bl\"ote, J.L.~Cardy  and M.P.~Nightingale,
\newblock {\em {Conformal invariance, the central charge, and universal
  finite-size amplitudes at criticality}},
\newblock Phys. Rev. Lett. {\textbf{56}} (1986)   742--745.

\bibitem{affleck:86}
I.~Affleck,
\newblock {\em Universal term in the free energy at a critical point and the
  conformal anomaly},
\newblock Phys. Rev. Lett. {\textbf{56}} (1986)   746--748.

\bibitem{diehl:81}
{H.~W.}~Diehl and S.~Dietrich,
\newblock {\em Field-theoretical approach to static critical phenomena in
  semi-infinite systems},
\newblock Z. Phys. B Cond. Mat. {\textbf{42}}
  {\textbf{1}} (1981)   65--86.

\bibitem{banchi:09}
L.~{Banchi}, F.~{Colomo}  and P.~{Verrucchi},
\newblock {\em {When finite-size corrections vanish: The $S= 1/2$ XXZ model
  and the Razumov-Stroganov state}},
\newblock Phys. Rev. A {\textbf{80}} (2009)   022341.

\bibitem{beccaria:12}
M.~Beccaria and C.~Hagendorf,
\newblock {\em A staggered fermion chain with supersymmetry on open intervals},
\newblock J. Phys. A: Math. Theor. {\textbf{45}} (2012)   365201.

\end{thebibliography}
\end{document}